\documentclass[11pt,a4paper]{article}

\usepackage{latexsym}
\usepackage{amsfonts}
\usepackage{amssymb}
\usepackage{amsmath}
\usepackage{wasysym}
\usepackage[mathscr]{eucal}
\usepackage{theorem}
\usepackage{enumerate}
\usepackage{cite}
\usepackage{url}
\usepackage{graphicx}
\usepackage[all]{xy}

\theoremstyle{plain}
\theorembodyfont{\itshape}
\newtheorem{thm}{Theorem}

\newtheorem{prp}{Proposition}

\newtheorem{dfn}{Definition}

\theorembodyfont{\upshape}

\theorembodyfont{\upshape}
\newtheorem{exam}{Example}

\newcommand{\qed}{\hfill\mbox{\raggedright $\Box$}\medskip}

\newcommand{\mydate}{
 \ifcase\month \or
 January\or February\or March\or April\or May\or June\or
 July\or August\or September\or October\or November\or December\fi
 \space \number\year}

\newcommand{\smin}{\,\raisebox{0.06em}{${\scriptstyle \in}$}\,}

\newcommand{\smcirc}{{\scriptstyle \,\circ\,}}
\newcommand{\circledstar}{\mathbin{\mathchoice%
 {\ooalign{\hss$\displaystyle\ocircle$\hss\cr\hss$\displaystyle\star$\hss\cr}}%
 {\ooalign{\hss$\textstyle\ocircle$\hss\cr\hss$\textstyle\star$\hss\cr}}%
 {\ooalign{\hss$\scriptstyle\ocircle$\hss\cr\hss$\scriptstyle\star$\hss\cr}}%
 {\ooalign{\hss$\scriptscriptstyle\ocircle$\hss\cr%
           \hss$\scriptscriptstyle\star$\hss\cr}}}}
\newcommand{\bwedge}{\raisebox{0.2ex}{${\textstyle \bigwedge}$}}
\newcommand{\bvee}{\raisebox{0.2ex}{${\textstyle \bigvee}$}}

\newcommand{\mathslf}[1]{\ensuremath{\mbox{\slshape\textsf{#1}}}}

\newcommand{\dbwedge}[2]{\raisebox{0.2ex}{${\textstyle \bigwedge}$}
 \ensuremath{_{\scriptstyle #1}^{\raisebox{-0.2ex}{${\scriptstyle #2}$}}\,}}

\begin{document}

\title{\mbox{} \\
       Lie Groupoids in Classical Field Theory I: \\
       Noether's Theorem \\[4ex]}

\author{Bruno T.\ Costa$^{\,1}$, Michael Forger$^{\,1}$~and~%
        Luiz Henrique P.\ P\^egas$^{\,1,2}$ \\[2ex]}
\date{\normalsize
      $^{1\,}$ Instituto de Matem\'atica e Estat\'{\i}stica, \\
      Universidade de S\~ao Paulo, \\
      Caixa Postal 66281, \\
      BR--05315-970~ S\~ao Paulo, SP, Brazil \\[4mm]
      $^2\,$ Departamento de Matem\'atica, \\
      Universidade Federal do Paran\'a, \\
      Caixa Postal 19081, \\
      BR--81531-980 Curitiba, PR, Brazil
      }
\maketitle

\thispagestyle{empty}

\vspace*{\fill}

\begin{abstract}
\noindent
 In the two papers of this series, we initiate the development of a new
 approach to implementing the concept of symmetry in classical field theory,
 based on replacing Lie groups/algebras by Lie groupoids/algebroids, which
 are the appropriate mathematical tools to describe local \linebreak symmetries
 when gauge transformations are combined with space-time transformations.
 \linebreak
 Here, we outline the basis of the program and, as a first step, show how to
 (re)formulate Noether's theorem about the connection between symmetries
 and conservation laws in this approach.
\end{abstract}

\vspace*{\fill}

\begin{flushright}
 \parbox{12em}{
  \begin{center}
   Universidade de S\~ao Paulo \\
   RT-MAP-1501 \\
   Revised Version \\
   October 2017
  \end{center}
 }
\end{flushright}

\vspace*{\fill}

\pagebreak

\vspace*{\fill}

\begin{center}
 \sc{this page intentionally left blank}
\end{center}
\thispagestyle{empty}

\vspace*{\fill}

\pagebreak

\setcounter{page}{1}

\section{Introduction}

Symmetry is a fundamental concept of science deeply rooted in human
culture, as can be verified by contemplating, e.g., the monumental
collection of papers assembled in~\cite{Har1,Har2}, testifying the
innumerous ways in which it permeates all areas of knowledge.
In mathematics, it has been the driving force for the development
of group theory, which began in the $19^{\mathrm{th}}$ century
with the work of E.~Galois and S.~Lie, formalizing the idea that
symmetry transformations can be assembled into groups.
And in physics, it has in the course of the $20^{\mathrm{th}}$
century become one of the most influential guiding principles
for the development of new theories, used in a wide variety of
contexts, and now plays an important role in practically all areas,
including mechanics and field theory, both classical and quantum.

Correspondingly, and not surprisingly, symmetries in physics appear
in different variants. \linebreak
For example, they may act through transformations in 3-dimensional
physical space such as translations, rotations and reflections (spatial
symmetries), or in the case of relativistic physics where space and time
merge into a single space-time continuum, transformations in 4-dimensional
space-time such as Lorentz transformations (space-time symmetries), or
else they may act only through transformations in an abstract internal
space that has nothing to do with physical space or space-time, being
instead related to the dynamical variables of the theory under
consideration (internal symmetries).
Similarly, they can be distinguished according to whether the
corresponding group of transformations is continuous, such as
the group of spatial translations or rotations (parametrized by
vectors or by the Euler angles, say), or is discrete, such as the
group of spatial translations or rotations in a crystal lattice or
a reflection group.
And finally, a very important extension of the usual symmetry
concept arises from the idea that the parameters characterizing a
specific element within the pertinent group may depend on the point
in space or space-time where the symmetry transformation is performed.
This notion of ``gauging a symmetry'', thus extending it from a global
symmetry to a local symmetry by allowing different transformations
to be performed at different points, is at the heart of gauge theories,
which occupy a central position in modern field theory, classical as
well as quantum.
(We think of this as an extension of the symmetry concept because
gauge transformations do not really represent symmetry transformations
in the strict sense of the word, since they do not relate observable quantities.
Instead, their presence reflects the fact that the system under consideration
is being described in terms of redundant variables which are not observable,
such as the potentials in electrodynamics, and the amount of redundancy
in the choice of these variables is controlled by the principle of gauge
invariance: the observable, physical content of the theory is encoded
in its gauge invariant part.)

The mathematical implementation of different types of symmetries
requires different types of groups.
In particular, continuous symmetries are described in terms of Lie
groups and, infinitesimally, of Lie algebras, and local symmetries
involve infinite-dimensional Lie groups and Lie algebras.
But more exotic mathematical entities have also begun to appear
on the scene, essentially since the 1970s, and have come to be
associated with generalized notions of symmetry, such as Lie
superalgebras and quantum groups.

In this paper, we propose to employ yet another mathematical tool
that is particularly well adapted to the concept of a local symmetry,
namely that of Lie groupoids and, infinitesimally, of Lie algebroids.
The main advantage of such an approach is that it eliminates the
need for working from the very beginning with infinite-dimensional
objects such as diffeomorphism groups of manifolds and automorphism
groups of bundles, whose mathematical structure is notoriously difficult
to handle: such groups do arise along the way but are now derived from
a more fundamental underlying entity, which is purely finite-dimensional.

This procedure of ``reduction to finite dimensions'' is by no means new
and has in fact been used in differential geometry for a long time, namely
when one considers a representation of a Lie group $G$ by diffeomorphisms
of a manifold~$M$ not as a group homomorphism $\, G \longrightarrow
\mathrm{Diff}(M)$, whose continuity (let alone smoothness) is hard
to define and even harder to control, but rather as a group action
$\, G \times M \longrightarrow M$, for which the definition of
continuity (and smoothness) is the obvious one.
A similar procedure can be applied when $G$ is replaced by one of the
aforementioned infinite-dimensional groups which appear naturally in
differential geometry and in gauge theories, the main difference being
that in this case the corresponding action is no longer that of a Lie
group on a manifold but rather that of a Lie groupoid on a fiber bundle.

The history of the theory of groupoids is to a certain extent parallel
to that of group theory itself. The notion of an abstract groupoid goes
back to a paper by H.~Brandt in 1927~\cite{Bra}; it applies to discrete
as well as to continuous symmetries.
Lie groupoids seem to have been introduced by C.~Ehresmann in the
1950s, in conjunction with principal bundles and connections, but in
contrast to these did not find their way into mainstream differential
geometry for several decades.
Lie algebroids and their relation to Lie groupoids were apparently 
first discussed by J.~Pradines in 1968~\cite{Pra}, but the main
result (Lie's third theorem) as stated there is incorrect and was
only rectified much later~\cite{CF}.
At the time of this writing, the standard textbook in the area is
Ref.~\cite{Mac}, whose results we shall use freely in this paper.

We conclude this introduction with an outline of the contents.
In Section~2, we briefly review the modern formulation of classical
field theory in a geometric framework, where fields are sections of
fiber bundles over space-time and hence coordinate invariance and
gauge invariance (in the sense of invariance under changes of local
trivializations of these fiber bundles) are built in from the very
beginning.
As in the standard geometric formulation of classical mechanics~%
\cite{AM,Arn}, we shall for the sake of simplicity restrict ourselves
to a first order formalism~-- lagrangian as well as hamiltonian.
In Section~3, we discuss a few basic topics from the theory of
Lie groupoids and Lie algebroids.
In particular, we present the construction of the jet groupoid
of a Lie groupoid, which is badly neglected in the mathematical
literature (for instance, it does not appear at all in Ref.~\cite{Mac})
but which turns out to be of crucial importance for the entire theory;
it can be iterated to produce a construction of the second order jet
groupoid of a Lie groupoid and of a non-holonomous extension
thereof which is needed in the sequel.
Apart from that, our goal in these two introductory sections
is essentially just to delineate the concepts we shall be using
and to fix the notation.
In Section~4, we discuss how an action of a Lie groupoid on
a fiber bundle induces actions of certain Lie groupoids derived
from the original one on certain fiber bundles derived from the
original one.
Here, the central point is that apart from obvious functorial
procedures such as defining the action of the jet groupoid on the
jet bundle or the tangent groupoid on the tangent bundle (of the
total space), we are also able to construct a new~-- and much
less obvious~-- induced action of the jet groupoid on the tangent
bundle (of the total space).
This construction is perhaps the most important result of the
paper because it allows us to give a precise mathematical
definition of invariance, under the action of a Lie groupoid
on some fiber bundle, of geometric structures on its total
space, at least when these are defined by some kind of
tensor field,  thus finally overcoming one of the major
obstacles in the theory that has for a long time jeopardized
its relevance for applications. \linebreak
The new feature as compared to group actions is that such
an invariance refers not to the original Lie groupoid itself
but rather to its jet groupoid, or some subgroupoid thereof.
We also exhibit the relation with the corresponding
representations of the appropriate groups of bisections and,
as an application, show in which sense the multicanonical
form $\theta$ and the multisymplectic form $\omega$ of
the covariant hamiltonian formalism, as presented in
Section~2, are invariant under the appropriate induced
actions.
In Section~5, we introduce the concept of momentum map
and prove Noether's theorem in this setting.
Finally, in Section~6, we illustrate some of the main concepts
and constructions introduced in the paper on the simplest
possible example: the theory of a single real scalar field.

In the second paper of this series, we shall specialize the
general formalism developed below to what may be regarded
as the most important general class of geometric field theories:
gauge theories.
There, all bundles that appear are derived from a given principal
bundle over space-time, either as an associated bundle (whose
sections are matter fields) or as the corresponding connection
bundle (whose sections are connections, or in physics language,
gauge potentials), and the Lie groupoid acting on them is the
corresponding gauge groupoid.
In this context, we shall then be able to discuss issues such as
the significance of the procedure of ``gauging a symmetry'',
already mentioned above, the prescription of minimal coupling
and Utiyama's theorem, generalizing the results of Ref.~\cite{FS}
from the context of Lie group bundles, which are sufficient to
handle internal symmetries, to Lie groupoids, as required to
deal with the general case of space-time symmetries mixed
with internal symmetries.

\section{Geometric formulation of classical field theory}

We start out by fixing a fiber bundle $E$ over a base manifold $M$,
with bundle projection denoted by $\, \pi_E^{}: E \longrightarrow M$:
it will be called the \emph{configuration bundle} since its sections
are the basic fields of the theory under consideration. This requires
that, in physics language, $M$ is to be interpreted as space-time~--
even though we do not assume it to carry any fixed metric, given
that in general relativity the metric tensor is itself a dynamical
variable and hence cannot be fixed ``a priori''.

In order to formulate the laws governing the dynamics of the fields,
we need to consider derivatives (velocities), as well as their duals
(momenta).

In order to do so, we begin by introducing the \emph{jet bundle} $JE$
of~$E$, together with the \emph{linearized jet bundle} $\vec{J} E$
of~$E$, as follows: for any point $e$ in~$E$ with base point
$\, x = \pi_E^{}(e) \,$ in~$M$, let $L(T_x^{} M,T_e^{} E)$
denote the space of linear maps from the tangent space
$T_x^{} M$ to the tangent space $T_e^{} E$ and consider
the affine subspace
\begin{equation} \label{eq:FJB1}
 J_e^{} E~=~\{ \, u_e^{} \smin L(T_x^{} M,T_e^{} E) \, | \;
               T_x^{} \pi_E^{} \smcirc u_e^{}
               = \mathrm{id}_{T_x^{} M}^{} \, \} \,,
\end{equation}
and its difference vector space
\begin{equation} \label{eq:LJB1}
 \vec{J}_e^{} E~=~\{ \, \vec{u}_e^{} \smin L(T_x^{} M,T_e^{} E) \, | \;
                     T_x^{} \pi_E^{} \smcirc \vec{u}_e^{} = 0 \, \} \,,
\end{equation}
i.e.,
\begin{equation} \label{eq:LJB2}
 \vec{J}_e E~=~L(T_x^{} M,V_e^{} E)~=~T_x^* M \otimes V_e^{} E \,,
\end{equation}
where $\, V_e^{} E = \ker T_e^{} \pi_E^{} \,$ is the vertical space
of~$E$ at~$e$.
Taking the disjoint union as $e$ varies over~$E$, this defines $JE$ and
$\vec{J} E$ as bundles in two different ways, which will collectively be
referred to as \emph{jet bundles}\/: over~$E$, $JE$ is an affine bundle
and $\vec{J} E$ is a vector bundle with respect to the corresponding
jet target projections $\, \pi_{JE}^{}: JE \longrightarrow E \,$ and
$\, \pi_{\vec{J} E}: \vec{J} E \longrightarrow E$, while over~$M$,
both of them are fiber bundles with respect to the corresponding jet
source projections (obtained from the former by composition with the
original bundle projection $\pi_E^{}$).%
\footnote{Here and throughout this paper, we face the problem that the
same expressions ``source'' and ``target'' are used in the theory of jets
and in the theory of groupoids, with different meanings. We shall avoid
confusion by adhering to the convention to use the prefix ``jet'' in the
first case.}
Moreover, composition with the appropriate tangent maps provides a
canonical procedure for associating with every strict homomorphism
$\, f: E \longrightarrow F \,$ of fiber bundles $E$ and~$F$ over~$M$ a
map $\, Jf: JE \longrightarrow JF$ (sometimes called its jet prolongation
or jet extension), which is a homomorphism of affine bundles (i.e.,
a fiberwise affine smooth map) covering~$f$, together with a map
$\, \vec{J} f: \vec{J} E \longrightarrow \vec{J} F$, \linebreak which
is a homomorphism of vector bundles (i.e., a fiberwise linear smooth
map) covering~$f$; \linebreak in particular, both are again strict
homomorphisms of fiber bundles over~$M$, so in fact $J$ and
$\vec{J}$ are \emph{functors} in the category of fiber bundles
over a fixed base manifold.
This is summarized in the following commuting diagrams:
\begin{equation} \label{eq:JFUNC1}
 \begin{array}{cc}
  \xymatrix{
   JE \ar[rr]^{Jf} \ar[d] & & JF \ar[d] \\
   E \ar[rr]^f \ar[dr]_{\pi_E^{}} & & F \ar[dl]^{\pi_F^{}} \\
   & M &
   }
  \quad & \quad
  \xymatrix{
   \vec{J} E \ar[rr]^{\vec{J} f} \ar[d] & & \vec{J} F \ar[d] \\
   E \ar[rr]^f \ar[dr]_{\pi_E^{}} & & F \ar[dl]^{\pi_F^{}} \\
   & M &
   }
 \end{array}
\end{equation}
Explicitly, given $\, e \in E \,$ with $\, \pi_E^{}(e) = x$,
$u_e^{} \in J_e^{} E \subset L(T_x^{} M,T_e^{} E)$,
$\, \vec{u}_e^{} \in \vec{J}_e^{} E = L(T_x^{} M,V_e^{} E)$,
we have
\begin{equation} \label{eq:FJB2}
 J_e^{} f (u_e^{})~=~T_e^{} f \smcirc u_e^{} \quad , \quad
 \vec{J}_e^{} f (\vec{u}_e^{})~=~T_e^{} f \smcirc \vec{u}_e^{}~.
\end{equation}
In particular, for any section $\varphi$ of~$E$, we have
\begin{equation} \label{eq:FJB3}
 j(f \smcirc \varphi)~=~Jf \smcirc j\varphi~.
\end{equation}

The important role jet bundles play in differential geometry is largely
due to the fact that they provide the adequate geometric setting for
taking derivatives of sections.
More precisely, any section of~$E$, say $\varphi$, induces canonically
a section of $JE$ which is often called its \emph{jet prolongation} or
\emph{jet extension} and which we may denote by $j\varphi$, as in
much of the mathematical literature, or by $(\varphi,\partial\varphi)$,
to indicate that it contains all the information about the values of
$\varphi$ and of its first order derivatives at each point of~$M$. 
But this prolongation is really just a reinterpretation of the tangent
map $T\varphi$ to~$\varphi$, since $\varphi$ being a section
of~$E$ implies that, for any $\, x \in M$, $T_x^{} \varphi \in
J_{\varphi(x)}^{} E \subset L(T_x^{} M,T_{\varphi(x)}^{} E)$.
Obviously, given $\, e \in E \,$ with $\, \pi_E^{}(e) = x$, every
jet $\, u_e^{} \in J_e^{} E \,$ can be represented as the
derivative at~$x$ of some section $\varphi$ of~$E$ satisfying
$\, \varphi(x) = e$, i.e., we can always find $\varphi$ such
that $\, u_e^{} = T_x^{} \varphi$, but this does of course
not mean that every section of $JE$, as a fiber bundle over~$M$,
can be written as the jet prolongation of some section of~$E$:
those that can be so written are called \emph{holonomous},
and it is then easy to see that a section $\tilde{\varphi}$
of $JE$ will be holonomous if and only if $\, \tilde{\varphi}
= j\varphi \,$ where $\,  \varphi = \pi_{JE}^{} \smcirc
\tilde{\varphi}$.

In passing, we note that sections of~$JE$ not as a fiber bundle
over~$M$ but as an affine bundle over~$E$ also have an important
role to play: they correspond to \emph{connections} in $E$, realized
through their \emph{horizontal lifting map} (of tangent vectors).
And if we fix a connection $\Gamma: E \longrightarrow JE$, we
can introduce the notion of \emph{covariant derivative} of a
section $\varphi$: this is then a section of $\vec{J} E$ which
we may denote by $(\varphi,D\varphi)$ and which is defined
as the difference $\, T\varphi - \Gamma \smcirc \varphi$.
Summarizing, we may specify the statement at the beginning
of the previous paragraph by saying that the jet bundle and the
linearized jet bundle provide the adequate geometric setting for
taking ordinary (partial) derivatives and for taking covariant
derivatives of sections, respectively.

In order to handle second order derivatives, we shall need the
\emph{second order jet bundle} $J^2 E$: this can be constructed
either directly or else by iteration of the previous construction and
subsequent reduction, which occurs in two steps.
First, observe that the iterated first order jet bundle $J(JE)$
admits two natural projections to $JE$, namely, the standard
projection $\, \pi_{J(JE)}^{}: J(JE) \longrightarrow JE \,$ and
the jet prolongation $\, J\pi_{JE}^{}: J(JE) \longrightarrow JE \,$
of the standard projection $\, \pi_{JE}^{}: JE \longrightarrow E$,
explicitly defined as follows: for $\, e \in E$, $u_e^{} \in J_e^{} E \,$
and $\, u'_{u_e} \in  J_{u_e}^{}(JE)$,
\begin{equation} \label{eq:SOJB1}
 (\pi_{J(JE)})_{u_e}^{} (u'_{u_e}) = u_e^{} \,,
\end{equation}
whereas
\begin{equation} \label{eq:SOJB2}
 (J\pi_{JE})_{u_e}^{} (u'_{u_e})
 = T_{u_e}^{} \pi_{JE}^{} \,\smcirc\, u'_{u_e} \,.
\end{equation}
They fit into the following commutative diagram:
\begin{equation} \label{eq:SOJB3}
 \begin{array}{c}
  \xymatrix{
   & J(JE) \ar[dl]_{\pi_{J(JE)}^{}} \ar[dr]^{J\pi_{JE}^{}} & \\
   JE \ar[dr]_{\pi_{JE}^{}} & & JE \ar[dl]^{\pi_{JE}^{}} \\
   & E &
  }
 \end{array}
\end{equation}
The first step of the reduction mentioned above is then to restrict to the 
subset of $J(JE)$ where the two projections coincide: for $\, e \in E \,$ 
and $\, u_e^{} \in J_e^{} E$, define
\begin{equation} \label{eq:SOJB4}
 \bar{J}_{u_e}^{\,2} E~
 =~\{ \, u'_{u_e} \smin J_{u_e}^{}(JE) \, | \;
         \pi_{J(JE)}^{}(u'_{u_e}) = J \pi_{JE}^{}(u'_{u_e}) \, \} \,.
\end{equation}
Taking the disjoint union as $u_e^{}$ varies over~$JE$, this defines what
we shall call the \emph{semiholonomic second order jet bundle} of~$E$,
denoted by $\bar{J}^{\,2} E$: it is naturally a fiber bundle over~$M$ and,
since $\pi_{J(JE)}$ and $J\pi_{JE}$ are both homomorphisms of affine
bundles, it is also an affine bundle over~$JE$.
The second step consists in decomposing this, as a fiber product of affine
bundles over~$JE$, into a symmetric part and an antisymmetric part:
the former is precisely $J^2 E$ and is an affine bundle over~$JE$, with
difference vector bundle equal to the pull-back to~$JE$ of the vector
bundle $\, \pi_E^* \bigl( \bvee^{2\,} T^\ast M \bigr) \otimes VE \,$
over~$E$ by the jet target projection $\pi_{JE}^{}$, whereas the
latter is a vector bundle over~$JE$, namely the pull-back to~$JE$
of the vector bundle $\, \pi_E^* \bigl( \bwedge^{\!2\,} T^\ast M \bigr)
\otimes VE \,$ over~$E$ by the jet target projection~$\pi_{JE}^{}$:
\begin{equation} \label{eq:SOJB6}
 \begin{array}{c}
  \bar{J}^2 E~\cong~J^2 E \; \times_{JE}^{} \; \pi_{JE}^*
                    \Bigl( \pi_E^* \bigl( \bwedge^{\!2\,} T^\ast M \bigr)
                                      \otimes VE \Bigr) \,, \\[2mm]
  \vec{J^2} E~\cong~\pi_{JE}^*
                    \Bigl( \pi_E^* \bigl( \bvee^{2\,} T^\ast M \bigr)
                                      \otimes VE \Bigr) \,.
 \end{array}
\end{equation}
But what is perhaps even more important is the following observation:
given any section $\tilde{\varphi}$ of $JE$, its jet prolongation
$j\tilde{\varphi}$ will be a section of $J(JE)$ which will take values
in $\bar{J}^{\,2} E$ if and only if $\tilde{\varphi}$ is holonomous.
(Indeed, let us set $\, \varphi = \pi_{JE}^{} \smcirc \tilde{\varphi}$.
If $\tilde{\varphi}$ is holonomous, then as already observed above, we
must have $\, \tilde{\varphi} = j\varphi \,$ and hence $\, J\pi_{JE}^{}
\smcirc j\tilde{\varphi} = j(\pi_{JE}^{} \smcirc \tilde{\varphi}) =
j\varphi = \tilde{\varphi} = \pi_{J(JE)}^{} \smcirc j\tilde{\varphi}$,
i.e., \linebreak $j\tilde{\varphi} = j^2 \varphi \,$ takes values in
$\bar{J}^{\,2} E$. Conversely, if $j\tilde{\varphi}$ takes values
in $\bar{J}^{\,2} E$, then $\, j\varphi = j(\pi_{JE}^{} \smcirc
\tilde{\varphi}) = J\pi_{JE}^{} \smcirc j\tilde{\varphi} =
\pi_{J(JE)}^{} \smcirc j\tilde{\varphi} = \tilde{\varphi}$.)
Of course, this means exactly that $\, j\tilde{\varphi} = j^2 \varphi$
is a section of~$J^2 E$, so we may characterize $J^2 E$ as the
unique ``maximal holonomous subbundle'' of $\bar{J}^{\,2} E$,
i.e., that subbundle of $\bar{J}^{\,2} E$ whose sections are precisely
the holonomous sections of $\bar{J}^{\,2} E$. For more details, see
\cite[Chapter~5]{Sau}.

Returning to the first order formalism, the next step consists in taking
duals. Briefly, the \emph{affine dual} $J^\star E$ of~$JE$ and the
\emph{linear dual} $\vec{J}^{\,\ast} E$ of~$\vec{J} E$ are defined as
follows: for any point $e$ in~$E$ with base point $\, x = \pi_E^{}(e) \,$
in~$M$, put
\begin{equation} \label{eq:ODFJB1}
 J_e^\star E~=~\{ \, z_e^{}: J_e^{} E \longrightarrow \mathbb{R} \, | \;
                            z_e^{}~\mbox{is affine} \, \} \,,
\end{equation}
and
\begin{equation} \label{eq:ODLJB1}
 \vec{J}_e^{\,\ast} E~
 =~\{ \, \vec{z}_e^{}: \vec{J}_e^{} E \longrightarrow \mathbb{R} \, | \;
                       \vec{z}_e^{}~\mbox{is linear} \, \} \,.
\end{equation}
However, the multiphase spaces of field theory are defined with an
additional twist, which consists in replacing the real line by the
one-dimensional space of volume forms on the base manifold $M$ at
the appropriate point. In other words, the \emph{twisted affine dual}
$J^{\circledstar} E$ of~$JE$ and the \emph{twisted linear dual}
$\vec{J}^{\,\circledast} E$ of~$\vec{J} E$ are defined as follows:
for any point $e$ in~$E$ with base point $\, x = \pi_E^{}(e) \,$
in~$M$, put
\begin{equation} \label{eq:TDFJB1}
 J_e^{\circledstar} E~
 =~\{ \, z_e^{}: J_e^{} E \longrightarrow
         \bwedge^{\!n\,} T_x^* M \, | \;
         z_e^{}~\mbox{is affine} \, \} \,,
\end{equation}
and
\begin{equation} \label{eq:TDLJB1}
 \vec{J}_e^{\,\circledast} E~
 =~\{ \, \vec{z}_e^{}: \vec{J}_e^{} E \longrightarrow
         \bwedge^{\!n\,} T_x^* M \, | \;
         \vec{z}_e^{}~\mbox{is linear} \, \} \,.
\end{equation}
Taking the disjoint union as $e$ varies over~$E$, this defines $J^\star E$,
$\vec{J}^{\,\ast} E$, $J^{\circledstar} E$ and $\vec{J}^{\,\circledast} E$
as bundles in two different ways, which will collectively be referred to
as \emph{cojet bundles}\/: all of them are vector bundles over~$E$ with
respect to the corresponding cojet target projections and fiber bundles
over~$M$ with respect to the corresponding cojet source projections
(obtained from the former by composition with the original bundle
projection $\pi_E^{}$).
Considered as vector bundles over $E$, we have
\begin{equation} \label{eq:TDFJB2}
 J^{\circledstar} E~
 =~J^\star E \otimes
   \pi_E^* \bigl( \bwedge^{\!n\,} T^\ast M \bigr) \,,
\end{equation}
and
\begin{equation} \label{eq:TDLJB2}
 \vec{J}^{\,\circledast} E~
 =~\vec{J}^{\,\ast} E \otimes
   \pi_E^* \bigl( \bwedge^{\!n\,} T^\ast M \bigr) \,.
\end{equation}
Moreover, $J^\star E$ is also an affine line bundle over $\vec{J}^{\,\ast} E$
and, similarly, $J^{\circledstar} E$ is also an affine line bundle over
$\vec{J}^{\,\circledast} E$, whose projections are defined, over each
point $e$ of $E$, by taking the linear part of an affine map.
In the twisted case, this bundle projection
\begin{equation} \label{eq:TDPROJ}
 \eta: J^{\circledstar} E~\longrightarrow~\vec{J}^{\,\circledast} E
\end{equation}
plays an important role because the \emph{hamiltonian} of any classical
field theory whose field content is captured by the configuration bundle~$E$
is a section
\begin{equation} \label{eq:HAMILT}
 \mathcal{H}: \vec{J}^{\,\circledast} E~\longrightarrow~J^{\circledstar} E
\end{equation}
of this projection~\cite{CCI}.
In a more physics oriented language, we call $J^{\circledstar} E$ the
\emph{extended multiphase space} and $\vec{J}^{\,\circledast} E$ the
\emph{ordinary multiphase space} associated with the given configuration
bundle~$E$.

Here, the term ``multiphase space'' is supposed to indicate
that the (twisted) cojet bundles $J^{\circledstar} E$ and
$\vec{J}^{\,\circledast} E$ carry a multisymplectic structure
which is analogous to the symplectic structure on cotangent
bundles that qualifies these as candidates for a ``phase space''
in classical mechanics.
Its definition relies on an immediate generalization of the
construction of the canonical $1$-form on a cotangent bundle
known from classical mechanics, namely, the fact that the bundle
$\bwedge^{\!r\,} T^\ast E$ of $r$-forms on any manifold $E$
carries a naturally defined ``tautological'' $r$-form $\theta$,
explicitly given~by
\begin{equation} \label{eq:TAUTF}
 \begin{array}{c}
  \theta_{\alpha} \bigl( w_1^{} , \ldots w_r^{} \bigr)~
  =~\alpha \bigl( T_\alpha \pi (w_1^{}) , \ldots ,
                       T_\alpha \pi (w_r^{}) \bigr) \\[1mm]
  \mbox{for $\, \alpha \in \bwedge^{\!r\,} T^\ast E \,,\,
                w_1,\ldots,w_r^{} \in
                T_\alpha \bigl( \bwedge^{\!r\,} T^\ast E \bigr)$} \,,
 \end{array}
\end{equation}
where $\pi$ denotes the bundle projection from $\bwedge^{\!r\,}
T^\ast E$ to~$E$, and that when $E$ is the total space of a fiber
bundle, then for $\, 0 \leqslant s \leqslant r$, this form restricts
to a ``tautological'' $(r-s)$-horizontal $r$-form $\theta$ on the
bundle $\dbwedge{s}{r} T^\ast E$ of $(r-s)$-horizontal
$r$-forms on~$E$,%
\footnote{An $r$-form on the total space of a fiber bundle is
said to be $(r-s)$-horizontal if it vanishes whenever one inserts
at least $s+1$ vertical vectors.}
which we shall continue to denote by $\theta$, combined with
the following canonical isomorphism of vector bundles over~$E$
\cite{GIM}:
\begin{equation} \label{eq:ISOJST1}
 J^{\circledstar} E~\cong~\dbwedge{1}{n} T^\ast E \,.
\end{equation}
This isomorphism allows us to transfer the form $\theta$, as well
as its exterior derivative (up to sign), $\omega = - d\theta$, to
forms on $J^{\circledstar} E$ which, for the sake of simplicity of
notation, we shall continue to denote by $\theta$ and by $\omega$,
respectively: then $\theta$ is called the \emph{multicanonical form}
and $\omega$ is called the \emph{multisymplectic form} on the
extended multiphase space $J^{\circledstar} E$.
Finally, we introduce the \emph{multicanonical form} $\theta_{\mathcal{H}}$
and the \emph{multisymplectic form} $\omega_{\mathcal{H}}$
on the ordinary multiphase space $\vec{J}^{\,\circledast} E$ by
pulling them back with the hamiltonian $\mathcal{H}$ (see equation~%
(\ref{eq:HAMILT}) above):
\begin{equation} \label{eq:MULTCS}
 \theta_{\mathcal{H}} = \mathcal{H}^* \theta \quad , \quad
 \omega_{\mathcal{H}} = \mathcal{H}^* \omega \,,
\end{equation}
noting that we still have $\, \omega = - d\theta \,$ and
$\, \omega_{\mathcal{H}} = - d\theta_{\mathcal{H}}$.

Within the context outlined above, the traditional method for fixing
the dynamics of a specific field theoretical model is by exhibiting
its \emph{lagrangian}, which is a homomorphism
\begin{equation} \label{eq:LAGRAN}
 \mathcal{L}: JE~\longrightarrow~\bwedge^{\!n\,} T^* M
\end{equation}
of fiber bundles over~$M$.
The hypothesis that $\mathcal{L}$, when composed with the jet
prolongation of a section $\varphi$ of~$E$, provides an $n$-form
on~$M$, rather than a function, allows us to define the \emph%
{action functional} $S$ directly by setting, for any compact
subset $K$ of~$M$,
\begin{equation} \label{eq:ACTION}
 S_K[\varphi]~=~\int_K \mathcal{L}(\varphi,\partial\varphi) \,,
\end{equation}
without the need of choosing a volume form on space-time: this
is supposed to be absorbed in the definition of the lagrangian.
Taking its fiber derivative gives rise to the \emph{Legendre
transformation}, which comes in two variants: as a homomorphism
\begin{equation} \label{eq:LEGT1}
 \mathbb{F} \mathcal{L}: JE~\longrightarrow~
                         J^{\circledstar} E
\end{equation}
or as a homomorphism
\begin{equation} \label{eq:LEGT2}
 \vec{\mathbb{F}} \mathcal{L}: JE~\longrightarrow~
                               \vec{J}^{\,\circledast} E
\end{equation}
of fiber bundles over~$E$. Explicitly, given $\, e \in E \,$
with $\, \pi_E^{}(e) = x \,$ and $u_e^{} \in J_e^{} E \subset
L(T_x^{} M,T_e^{} E)$, the latter is defined as the usual fiber
derivative of $\mathcal{L}$ at $u_e^{}$, which is the linear map
from $\vec{J}_e^{} E$ to $\bwedge^{\!n\,} T_x^* M$ given by
\begin{equation} \label{eq:LEGT3}
 \vec{\mathbb{F}} \mathcal{L}(u_e^{}) \cdot \vec{u}_e^{}~
 =~\frac{d}{dt} \, \mathcal{L}(u_e^{} +t \vec{u}_e^{}) \>\!
   \Big|_{t=0} \qquad
 \mbox{for $\, \vec{u}_e^{} \in \vec{J}_e^{} E = L(T_x^{} M,V_e^{} E)$} \,,
\end{equation}
whereas the former encodes the entire Taylor expansion, up to first order,
of $\mathcal{L}$ around~$u_e^{}$ along the fibers, which is the affine
map from $J_e^{} E$ to $\bwedge^{\!n\,} T_x^* M$ given by
\begin{equation} \label{eq:LEGT4}
 \mathbb{F} \mathcal{L}(u_e^{}) \cdot u_e'~
 =~\mathcal{L}(u_e^{}) \, + \,
   \frac{d}{dt} \, \mathcal{L}(u_e^{} + t (u_e' - u_e^{})) \>\!
   \Big|_{t=0} \qquad
 \mbox{for $\, u_e' \in J_e^{} E \subset L(T_x^{} M,T_e^{} E)$} \,.
\end{equation}
Of course, $\vec{\mathbb{F}} \mathcal{L}$ is just the linear part
of $\mathbb{F} \mathcal{L}$, so that it is simply its composition
with the bundle projection $\eta$ from extended to ordinary
multiphase space: $\vec{\mathbb{F}} \mathcal{L} = \eta \circ
\mathbb{F} \mathcal{L}$.
Conversely, we require the hamiltonian $\mathcal{H}$ (see equation~%
(\ref{eq:HAMILT}) above) to satisfy the same sort of composition rule,
but in the opposite direction: $\mathbb{F} \mathcal{L} = \mathcal{H}
\circ \vec{\mathbb{F}} \mathcal{L}$.
In particular, if the lagrangian $\mathcal{L}$ is supposed to be
\emph{hyperregular}, which by definition means that $\vec{\mathbb{F}}
\mathcal{L}$ should be a global diffeomorphism, then this condition can
be used to define the corresponding De Donder\,--\,Weyl hamiltonian
$\mathcal{H}$ as $\, \mathcal{H} = \mathbb{F} \mathcal{L} \smcirc
(\vec{\mathbb{F}} \mathcal{L})^{-1}$.

To conclude this discussion, we mention that fixing the configuration
bundle $E$ which has been our starting point here, the basic fields
of the theory will be sections $\varphi$ of $E$: by jet prolongation,
they give rise to sections $j\varphi = (\varphi,\partial\varphi)$
of~$JE$ and then, by composition with the Legendre transformation
$\vec{\mathbb{F}} \mathcal{L}$, to sections $\phi = (\varphi,\pi)$
of $\vec{J}^{\,\circledast} E$.
Moreover, the equations of motion that result from the principle of
stationary action, applied to the action functional as defined by
equation~(\ref{eq:ACTION}), can be formulated globally (i.e.,
without resorting to local coordinate expressions) as follows:
in the lagrangian framework, the section $\varphi$ should
satisfy the condition
\begin{equation} \label{eq:LEQMOT} 
 (\varphi,\partial\varphi)^* (i_X^{} \, \omega_{\mathcal{L}}^{})~=~0
\end{equation}
for any vector field $X$ on $JE$ which is vertical under the
projection to~$M$, or even for any vector field $X$ on $JE$
which is projectable to~$M$, where
\begin{equation} \label{eq:POCARF} 
 \omega_{\mathcal{L}}^{}~
 =~(\vec{\mathbb{F}} \mathcal{L})^* \omega_{\mathcal{H}}^{}~
 =~({\mathbb{F}} \mathcal{L})^* \omega 
\end{equation}
is the \emph{Poincar\'e\,--\,Cartan form} (Euler\,--\,Lagrange equations),
whereas in the hamiltonian framework, the section $\phi$ should satisfy
the condition
\begin{equation} \label{eq:HEQMOT} 
 \phi^* (i_X^{} \, \omega_{\mathcal{H}}^{})~=~0
\end{equation}
for any vector field $X$ on $\vec{J}^{\,\circledast} E$ which is
vertical under the projection to~$M$, or even for any vector field
$X$ on $\vec{J}^{\,\circledast} E$ which is projectable to~$M$
(De~Donder\,--\,Weyl equations).

\noindent
Graphically, we can visualize the situation in terms of the following
commutative diagram:
\begin{equation} \label{eq:DIAGR1}
 \begin{minipage}{5cm}
  \xymatrix{
   & ~J^{\circledstar} E~ \\
   ~JE~ \ar[ru]^{\mathbb{F}\mathcal{L}}
        \ar[r]_{\vec{\mathbb{F}}\mathcal{L}~~} \ar[d]^{\pi_{\!JE}^{}}
   & ~\vec{J}^{\,\circledast} E~ \ar@<0.5ex>[u];[]^(.54){\eta}
                                 \ar@<0.3ex>[u]^(.5){\mathcal{H}} \ar[ld] \\
   E \ar[d]^{\pi_{\!\!\>E}^{}} & \\
   M \ar@/^/[u]|{\varphi\rule[-0.7ex]{0em}{2ex}}
     \ar@/^2pc/[uu]|{j\varphi\rule[-1ex]{0em}{2.5ex}}
     \ar@/_/[uur]|{\phi\rule[-1ex]{0em}{3ex}} &
  }
 \end{minipage}
\end{equation}
A more detailed treatment of some aspects that, for the sake of
brevity, have been omitted here, including explicit expressions in
terms of local coordinates, can be found in Ref.~\cite{FSR}.

\section{Lie groupoids and Lie algebroids}

In this section, we shall briefly review the concept of a Lie groupoid,
as well as that of an action of a Lie groupoid on a fiber bundle, and
then present the construction of the jet groupoid of a Lie groupoid,
which is of central importance for applications to field theory.
We also discuss the infinitesimal versions of these concepts, that is,
Lie algebroids, infinitesimal actions and the construction of the jet
algebroid of a Lie algebroid. 
For details and proofs of many of these results, the reader is referred
to Ref.~\cite{Mac}: our main goal here is to simplify the notation.

\subsection{Lie groupoids}
\label{subsec:LGR}

The main feature distinguishing a (Lie) groupoid from a (Lie) group
is that, similar to a fiber bundle, it comes with a built-in base space,
but it carries two projections to that base space and not just one,
called the \emph{source projection} and the \emph{target
projection}, so that we can think of the elements of the groupoid
as transformations from one given point of the base space to another,
and composition of such transformations is allowed only when the
target of the first coincides with the source of the second.
Formally, a Lie groupoid \emph{over} a manifold $M$ is a
manifold\,%
\footnote{For some applications, it is necessary to allow the topology
of the total space $G$ not to be Hausdorff, but we shall not encounter
such types of Lie groupoids here.}
$G$ equipped with various structure maps which, when dealing with
various Lie groupoids at the same time, we shall decorate with an
index~$G$, namely, the \emph{source projection} $\, \sigma_G^{}:
G \longrightarrow M$ \linebreak and \emph{target projection}
$\, \tau_G^{}: G \longrightarrow M$, which are assumed to
be surjective submersions and are sometimes combined into
a single map $\, (\tau_G^{},\sigma_G^{}): G \longrightarrow
M \times M \,$ called the \emph{anchor}, the \emph{multiplication map}
\begin{equation} \label{eq:MULTM}
 \begin{array}{cccc}
  \mu_G^{}: & G \times_M G & \longrightarrow &  G \\
            &    (h,g)     &   \longmapsto   & h g
 \end{array}
\end{equation}
defined on the submanifold
\begin{equation} \label{eq:MULTD}
 G \times_M G~
 =~\bigl\{ (h,g) \in G \times G \,|\, \sigma_G^{}(h)=\tau_G^{}(g) \bigr\}
\end{equation}
of $G \times G$, the \emph{unit map}
\begin{equation} \label{eq:UNITM}
 \begin{array}{cccc}
  1_G: & M & \longrightarrow &  G  \\[1mm]
     & x &   \longmapsto   & (1_G)_x
 \end{array}
\end{equation}
and the \emph{inversion}
\begin{equation} \label{eq:INVER}
 \begin{array}{cccc}
  \iota_G^{}: & G & \longrightarrow &    G     \\
                & g &  \longmapsto   & g^{-1}
 \end{array}
\end{equation}
satisfying the usual axioms of group theory, whenever the
expressions involved are well-defined, such as the condition
of associativity $\, k (hg) = (kh) g \,$ for $k,h,g \in G$ such
that $\, \sigma_G^{}(k) = \tau_G^{}(h)$ \linebreak and
$\, \sigma_G^{}(h) = \tau_G^{}(g)$.
For the convenience of the reader, we quickly collect some
of the standard terminology in this area:
\begin{itemize}
 \item given $\, x \in M$, $G_x = \sigma_G^{-1}(x)$ is called
       the \emph{source fiber} over~$x$;
 \item given $\, y \in M$, ${}_y G = \tau_G^{-1}(y)$ is called
       the \emph{target fiber} over~$y$;
 \item given $\, x,y \in M$, ${}_y G_x = G_x \bigcap {}_y G \,$
       is called the \emph{joint fiber} over~$(y,x)$;
 \item given $\, x \in M$, ${}_x G_x = G_x \bigcap {}_x G \,$
       is a group, called the \emph{isotropy group} or
       \emph{stability group} or \emph{vertex group} at~$x$;
 \item given $\, x \in M$, $G \cdot x = \{ y \in M \,|\,
       {}_y G_x \neq \emptyset \} \,$ is the \emph{orbit} of~$x$;
 \item $G$ is said to be \emph{transitive} if there is a unique orbit,
       namely, all of $M$, that is, if for any $\, x \in M$, $G \cdot x = M$;
 \item $G$ is said to be \emph{totally intransitive} if the orbits are
       reduced to points, that is, if for any $\, x \in M$,
       $G\cdot x =\{x\}$, or equivalently, if the source
       and target projections coincide.
\end{itemize}
More generally, $G$ is said to be \emph{regular} if the anchor has
constant rank.
In this case, the subset of $G$ where $\sigma_G^{}$ and
$\tau_G^{}$ coincide,
\begin{equation} \label{eq:ISOTR}
 G_{\mathrm{iso}}^{}~
 =~\{ g \in G \,|\, \sigma_G^{}(g) = \tau_G^{}(g) \} \,,
\end{equation}
is a totally intransitive Lie subgroupoid of $G$, called its
\emph{isotropy subgroupoid}, and when it is locally trivial
(which may not always be the case), it will in fact be a
\emph{Lie group bundle} over~$M$.
Similarly, in this case, the orbits under the action of~$G$
provide a foliation of~$M$ called its \emph{characteristic
foliation}.

We shall also need the notion of morphism between Lie groupoids:
given two Lie groupoids $G$ and $G'$ over the same manifold~$M$,%
\footnote{For the sake of simplicity, we consider here only strict
morphisms, that is, morphisms between Lie groupoids over a fixed
base manifold~$M$ that induce the identity on~$M$. The general
case can be reduced to this one by employing the construction of
pull-back of Lie groupoids.}
a smooth map $\, f: G \longrightarrow G' \,$ is said to be a \emph%
{morphism}, or \emph{homomorphism}, of Lie groupoids if
$\, \sigma_{G'}^{} \smcirc f = \sigma_G^{}$, $\tau_{G'}^{}
\smcirc f = \tau_G^{} \,$ and
\begin{equation}
 f(hg) = f(h) f(g)
 \qquad \mbox{for $\, (h,g) \in G \times_M G$}~.
\end{equation}

\pagebreak

Of course, when $M$ reduces to a point, we recover the definition
of a Lie group.
At the other extreme, we have the following example:
\begin{exam}~ \label{ex:PAIRGR}
 Given a manifold $M$, consider the cartesian product $M \times M$
 of two copies of $M$ and define the projection onto the second/first
 factor as source/target projection ($\sigma(y,x) = x$, $\tau(y,x) = y$),
 juxtaposition with omission as multiplication ($(z,y)(y,x) = (z,x)$), the
 diagonal as unit ($1_x = (x,x)$) and switch as inversion ($(y,x)^{-1}
 = (x,y)$). Then $M \times M$ is a Lie groupoid over~$M$, called the
 \emph{pair groupoid} of~$M$.
\end{exam}
Here is another example that will be important in the sequel:
\begin{exam}~ \label{ex:ONFRGR}
 Given a manifold $M$, consider its tangent bundle $TM$ and set
 \[
  GL(TM)~=~\dot{\bigcup_{x,y \in M}} \; {}_y GL(TM)_x
  \qquad \mbox{where} \qquad
  {}_y GL(TM)_x = GL(T_x M,T_y M)
 \]
 is the set of invertible linear transformations from $T_x M$ to $T_y M$.
 Then $GL(TM)$, equipped with the obvious operations, is a Lie groupoid
 over~$M$, called the \emph{linear frame groupoid} of~$M$.
 Similarly, if in addition we fix a pseudo-riemannian metric $\mathslf{g}$
 on~$M$ and set
 \[
  O(TM,\mathslf{g})~=~\dot{\bigcup_{x,y \in M}} \; {}_y O(TM,\mathslf{g})_x
  \qquad \mbox{where} \qquad
  {}_y O(TM,\mathslf{g})_x = O((T_x M.\mathslf{g}_x),(T_y M,\mathslf{g}_y))
 \]
 is the set of orthogonal linear transformations from $(T_x M,\mathslf{g}_x)$
 to $(T_y M,\mathslf{g}_y)$, we arrive at what is called the \emph{orthonormal
 frame groupoid} of~$M$ (with respect to~$\mathslf{g}$).
\end{exam}
In fact, there is a wealth of examples of Lie groupoids, some
of which are direct generalizations of Lie groups (such as the
frame groupoids of Example~2) and some of which are not
(such as the pair groupoid of Example~1).
In particular, Example~2 is a special case of a more general
construction which associates to each principal bundle over~$M$ a
transitive Lie groupoid over~$M$ called its \emph{gauge groupoid},
here applied to the linear frame bundle and the orthonormal frame
bundle of~$M$, as a principal $GL(n,\mathbb{R})$-bundle and
$O(n)$-bundle, respectively, where $n = \dim M$.

The insertion of the groupoid concept into group theory becomes apparent
when we introduce the notion of bisection of a groupoid, which is the precise
analogue of the notion of section of a bundle, the only modification coming
from the fact that we now have to deal with two projections rather than one.
Namely, just as a (smooth) \emph{section} of a fiber bundle $E$ over a
manifold~$M$ is a (smooth) map $\, \varphi: M \longrightarrow E \,$ such
that $\, \pi_E^{} \smcirc \varphi = \mathrm{id}_M^{}$, a (smooth) \emph%
{bisection} of a Lie groupoid $G$ over a manifold~$M$ is a (smooth) map
$\, \beta: M \longrightarrow G \,$ such that $\, \sigma_G^{} \smcirc \beta
= \mathrm{id}_M^{}$ \linebreak and $\, \tau_G^{} \smcirc \beta \in
\mathrm{Diff}(M)$.
The point here is that the set $\mathrm{Bis}(G)$ of all bisections of a
Lie groupoid~$G$ is a group, with product defined by
\begin{equation} \label{eq:BISPR1}
 (\beta_2 \beta_1)(x)~=~\beta_2(\tau_G^{}(\beta_1(x))) \, \beta_1(x)
 \qquad \mbox{for $\, x \in M$},
\end{equation}
so that
\begin{equation} \label{eq:BISPR2}
 \tau_G^{} \smcirc (\beta_2 \beta_1)~
 =~(\tau_G^{} \smcirc \beta_2) \smcirc (\tau_G^{} \smcirc \beta_1),
\end{equation}
with unit given by the unit map of equation~(\ref{eq:UNITM}), which
clearly is a bisection, and with inversion defined by
\begin{equation} \label{eq:BISINV}
 \beta^{-1}(x)~
 =~\bigl( \beta \bigl( (\tau_G^{} \smcirc \beta)^{-1}(x) \bigr) \bigr)^{-1}
 \qquad \mbox{for $\, x \in M$}.
\end{equation}
Moreover, at least in the regular case (which is the only one of interest
here), $\mathrm{Bis}(G)$ has a natural subgroup, namely the group
$\mathrm{Bis}(G_{\mathrm{iso}}^{})$ of bisections (= sections) of
the isotropy groupoid of~$G$, and a natural quotient group, namely
the group $\mathrm{Diff}^G(M)$ of diffeomorphisms of~$M$
obtained by composition of bisections of~$G$ with the target
projection of~$G$, that is, the image of $\mathrm{Bis}(G)$
under the homomorphism
\begin{equation}
 \begin{array}{ccc}
  \mathrm{Bis}(G) & \longrightarrow &    \mathrm{Diff}(M)     \\[1mm]
       \beta      &   \longmapsto   & \tau_G^{} \smcirc \beta
 \end{array}
\end{equation}
which fit together into the following exact sequence of groups:
\begin{equation}
 \{1\}~\longrightarrow~\mathrm{Bis}(G_{\mathrm{iso}}^{})
      ~\longrightarrow~\mathrm{Bis}(G)~\longrightarrow
      ~\mathrm{Diff}^G(M)~\longrightarrow~\{1\} \,.
\end{equation}

One more concept of central importance that we shall introduce is
that of an \emph{action} of a Lie groupoid $G$ on a fiber bundle $E$,
both over the same base manifold~$M$: this is simply a map
\begin{equation} \label{eq:ACTLG1}
 \begin{array}{cccc}
  \Phi_E^{}: & G \times_M E & \longrightarrow &     E    \\
             &    (g,e)     &   \longmapsto   & g \cdot e
 \end{array}
\end{equation}
defined on the submanifold
\begin{equation} \label{eq:ACTLG2}
 G \times_M E~
 =~\bigl\{ (g,e) \in G \times E \,|\, \sigma_G^{}(g)=\pi_E^{}(e) \bigr\}
\end{equation}
of $G \times E$, satisfying $\, \pi_E^{} \smcirc \Phi_E^{} = \tau_G^{}
\smcirc \mathrm{pr}_1^{}$ (so that for any $\, g \in G$, left translation
by $g$ becomes a well-defined map from $E_{\sigma_G(g)}$ to
$E_{\tau_G(g)}$), together with the usual axioms of an action,
adapted to the requirement that the expressions involved should
be well-defined, such as the composition rule $\, h \cdot (g \cdot e)
= (hg) \cdot e \,$ for $h,g \in G$ and $e \in E$ such that
$\, \sigma_G^{}(h) = \tau_G^{}(g) \,$ and
$\, \sigma_G^{}(g) = \pi_E^{}(e)$. \linebreak
It is then easy to check that any such action induces a representation
of the group $\mathrm{Bis}(G)$ of bisections of~$G$ by automorphisms
of~$E$, that is, a homomorphism
\begin{equation} \label{eq:ACTREP1}
 \begin{array}{cccc}
  \Pi_E^{}: & \mathrm{Bis}(G) & \longrightarrow & \mathrm{Aut}(E)  \\[1mm]
            &      \beta      &   \longmapsto   & \Pi_E^{}(\beta)
 \end{array}
\end{equation}
defined by
\begin{equation} \label{eq:ACTREP2}
 \Pi_E^{}(\beta)~
 =~\Phi_E^{} \smcirc (\beta \smcirc \pi_E^{},\mathrm{id}_E^{}) \,.
\end{equation}
Note that the automorphism $\Pi_E^{}(\beta)$ of~$E$ covers the
diffeomorphism $\tau_G^{} \smcirc \beta$ of~$M$, so in particular
$\Pi_E^{}(\beta)$ will be a strict automorphism when $\beta$ is
a bisection (= section) of the isotropy groupoid of~$G$ and hence
we get the following commutative diagram,
\begin{equation} \label{eq:EXSEQ1}
 \begin{array}{c}
  \xymatrix{
   \{1\} \ar[r] & \mathrm{Aut}_s^{}(E) \ar[r] &
   \mathrm{Aut}(E) \ar[r] &
   \mathrm{Diff}^E(M) \ar[r] & \{1\} \\
   \{1\} \ar[r] & \mathrm{Bis}(G_{\mathrm{iso}}) \ar[r] \ar[u] &
   \mathrm{Bis}(G) \ar[r] \ar[u] &
   \mathrm{Diff}^G(M) \ar[r] \ar[u] & \{1\}
  }
 \end{array}
\end{equation}
where $\mathrm{Aut}_s^{}(E)$ denotes the group of strict automorphisms
of~$E$ and $\mathrm{Diff}^E(M)$ denotes the group of diffeomorphisms
of~$M$ admitting some lift to an automorphism of~$E$.


\subsection{The jet groupoid of a Lie groupoid}

There are various ways of constructing new Lie groupoids from a given one.
As examples, we may mention the \emph{tangent groupoid} and the \emph%
{action groupoid}, both of which are discussed and widely used in the
literature.
However, it must be emphasized that both of these constructions imply
a change of base manifold.
Briefly, given a Lie groupoid $G$ over~$M$, the tangent bundle~$TG$
of the total space~$G$ is again a Lie groupoid, but over the tangent
bundle $TM$ of the base space~$M$ rather than over $M$ itself: this
is easily seen by defining all structure maps of~$TG$ by applying
the tangent functor to the corresponding structure maps of~$G$.
Similarly, given a Lie groupoid $G$ and a fiber bundle $E$ over~$M$,
together with an action of~$G$ on~$E$, the submanifold $G \times_M E$
of $G \times E$ introduced above becomes a Lie groupoid, but over
the total space~$E$ of the given fiber bundle rather than over~$M$
(here, the source projection is just projection onto the first factor
while the target projection is the action map $\Phi_E^{}$ itself).
Mathematically, such a change of base space does not present any
serious problem, but it does jeopardize the physical interpretation
because the whole idea of localizing symmetries that underlies
the use of Lie groupoids in field theory refers to space-time and
not to any other manifold mathematically constructed from it.
Therefore, we should look for constructions that preserve the
base space.

This argument leads us directly to the construction of the jet
groupoid of a Lie groupoid, which is almost entirely analogous
to that of the jet bundle of a fiber bundle, involving just one
slight modification:
\begin{prp}~
 Given a Lie groupoid $G$ over a manifold $M$, set
 \[
  JG~=~\dot{\bigcup_{g \in G}} \; J_g G
 \]
 where, for $g \in G$ with $\sigma_G^{}(g) = x \,$ and
 $\, \tau_G^{}(g) = y$,
 \begin{equation} \label{eq:JETG1}
  J_g^{} G~=~\bigl\{ u_g^{} \in L(T_x^{} M,T_g^{} G) \,|\,
             T_g^{} \sigma_G^{} \smcirc u_g^{} = \mathrm{id}_{T_x M}^{}
             ~\mbox{and}~
             T_g^{} \tau_G^{} \smcirc u_g^{} \in GL(T_x^{} M,T_y^{} M) \bigr\}
 \end{equation}
 which is an open dense subset of the affine space as defined in
 equation~(\ref{eq:FJB1}), with $E$ replaced by~$G$, $\pi_E^{}$
 replaced by $\sigma_G^{}$ and $e$ replaced by $g$.%
 \footnote{The fact that we have slightly modified the meaning of
 the symbol $J$ when applying it to Lie groupoids, rather than to
 fiber bundles or, more generally, fibered manifolds, is unlikely to
 cause any confusion because we are dealing here with distinct
 categories: equation~(\ref{eq:FJB1}), with the aforementioned
 substitutions, would only apply if we were to consider $G$
 not as a Lie groupoid but just as a fibered manifold with
 respect to its source projection~-- which we do not.}
 Then $JG$ is again a Lie groupoid over~$M$, called
 the \textbf{jet groupoid} of~$G$, with source projection
 $\, \sigma_{JG}^{}: JG \longrightarrow M$, target projection
 $\, \tau_{JG}^{}:JG \longrightarrow M$, multiplication map
 $\, \mu_{JG}^{}: JG \times_M JG \longrightarrow JG$, unit
 map $\, 1_{JG}^{}: M \longrightarrow JG \,$ and inversion
 $\iota_{JG}^{}: JG \longrightarrow JG$ defined as follows:
 \begin{itemize}
  \item for $g \in G$ and $\, u_g^{} \in J_g^{} G$,
        \[
         \sigma_{JG}^{}(u_g) = \sigma_G^{}(g)~~,~~
         \tau_{JG}^{}(u_g) = \tau_G^{}(g) \,;
        \]
  \item for $g,h \in G$ with $\tau_G^{}(g) = \sigma_G^{}(h)$,
        $u_g^{} \in J_g^{} G \,$ and $\, v_h^{} \in J_h^{} G$,
        putting $\sigma_G^{}(g) = x$ and concatenating
        $v_h^{} \smcirc (T_g^{} \tau_G^{} \smcirc u_g^{})$
        and $u_g^{}$ into a single linear map
        $\bigl( v_h^{} \smcirc (T_g^{} \tau_G^{} \smcirc u_g^{}) ,
        u_g^{} \bigr)$ from $T_x^{} M$ into $\, T_h^{} G \oplus
        T_g^{} G \,$ and checking that this in fact takes values in
        the subspace \linebreak $T_{(h,g)}^{} (G \times_M G)$,
        \[
         v_h^{} u_g^{} \equiv \mu_{JG}^{}(v_h^{},u_g^{})
         = T_{(h,g)}^{} \mu_G^{} \smcirc
           \bigl( v_h^{} \smcirc (T_g^{} \tau_G^{} \smcirc u_g^{}) ,
                  u_g^{} \bigr) \,;
        \]
  \item for $x \in M$,
        \[
         (1_{JG}^{})_x^{} = T_x^{} 1_G^{} \,;
        \]
  \item for $g \in G$ and $\, u_g^{} \in J_g^{} G$,
        \[
         u_g^{-1} \equiv \iota_{JG}^{}(u_g)
         = T_g \iota_G^{} \smcirc u_g^{} \smcirc
           (T_g^{} \tau_G^{} \smcirc u_g^{})^{-1}.
        \]
        \vspace{-2ex}
 \end{itemize}
 Note that $JG$ is not only a Lie groupoid over~$M$ but also a fiber
 bundle over the total space of the original Lie groupoid~$G$.
\end{prp} 
We would like to emphasize that $JG$ does \emph{not} contain $G$ itself
as a Lie subgroupoid; rather, $G$ is a quotient of $JG$ under the natural
projection
\begin{equation} \label{eq:JETGP1}
 \begin{array}{cccc}
  \pi_{JG}^{}:
  &   JG   & \longrightarrow & G \\
  & u_g^{} &   \longmapsto   & g
 \end{array}
\end{equation}
which by construction is a morphism of Lie groupoids over~$M$.
Similarly, there is a natural projection
\begin{equation} \label{eq:JETGP2}
 \begin{array}{cccc}
  \pi_{JG}^{\mathrm{fr}}:
  &   JG   & \longrightarrow &              GL(TM)             \\
  & u_g^{} &   \longmapsto   & T_g^{} \tau_G^{} \smcirc u_g^{}
 \end{array}
\end{equation}
which again is a morphism of Lie groupoids over $M$.
Taking the fibered product $GL(TM) \times_M G$ of $GL(TM)$ and $G$
over $M$, which is again a Lie groupoid over~$M$, these two combine 
into a natural projection
\begin{equation} \label{eq:JETGP3}
 \begin{array}{cccc}
  \pi_{JG}^{\mathrm{fr}} \times_M^{} \pi_{JG}^{}:
  &   JG   & \longrightarrow &          GL(TM) \times_M G          \\
  & u_g^{} &   \longmapsto   & (T_g^{} \tau_G^{} \smcirc u_g^{},g)
 \end{array}
\end{equation}
which is another morphism of Lie groupoids over~$M$.

One can also make $J$ act on morphisms between Lie groupoids 
(over the same base manifold): given any two Lie groupoids $G$ and $G'$
over~$M$ and a morphism $\, f: G \longrightarrow G'$, one defines a
morphism $\, Jf: JG \longrightarrow JG' \,$ (sometimes called its jet
prolongation or jet extension) by setting
\begin{equation}
 J_g^{} f (u_g {})~=~T_g^{} f \smcirc u_g
 \qquad \mbox{for $\, g \in G, u_g^{} \in J_g^{} G$} \,.
\end{equation}
Obviously, $Jf$ covers $f$, and one has the following commutative diagram:
\begin{equation} \label{eq:JFUNC2}
 \begin{array}{c}
  \xymatrix{
   JG \ar[rr]^{Jf} \ar[d] & & JG' \ar[d] \\
   G \ar[rr]^f \ar@<-0.4ex>[dr]_-{\sigma_G^{}\hspace{-3pt}}
   \ar@<0.4ex>[dr]^>>>>>{\!\!\!\!\tau_{G^{\vphantom{\prime}}}^{}} 
   &   & G' \ar@<-0.4ex>[dl]_>>>>>{\sigma_{G^\prime}^{}\hspace*{-6pt}}
            \ar@<0.4ex>[dl]^-{\tau_{G'}^{}} \\
   & M &
  }
 \end{array}
\end{equation}
In this way, $J$ becomes a \emph{functor} in the category of Lie groupoids,
just as in the category of fiber bundles (always over a fixed base manifold).
\begin{exam}~ \label{ex:LFRGR}
 Up to a canonical isomorphism, the linear frame groupoid $GL(TM)$ is the
 jet groupoid of the pair groupoid $M \times M$,
 \[
  GL(TM) \cong J(M \times M) \,,
 \]
 and under this isomorphism, the projection in equation~(\ref{eq:JETGP2})
 above corresponds to the jet prolongation of the anchor map $\, (\tau_G^{},
 \sigma_G^{}): G \longrightarrow M \times M$.
\end{exam}

Just like jet bundles, jet groupoids can be expected to play an important
role in differential geometry because they provide the adequate geometric
setting for taking derivatives of bisections (which replace sections).
Namely, any bisection of~$G$, say~$\beta$, induces canonically a
bisection of~$JG$ which will be called its \emph{jet prolongation}
or \emph{jet extension} and which we may denote by $j\beta$ or
by $(\beta,\partial\beta)$: it is again just a reinterpretation
of the tangent map $T\beta$ to~$\beta$, since $\, \beta \in
\mathrm{Bis}(G)$ implies that, for any $\, x \in M$, $T_x \beta
\in J_{\beta(x)}^{} G$.
Obviously, given $\, g \in G \,$ with $\, \sigma_G^{}(g) = x$,
every jet $\, u_g^{} \in J_g^{} G \,$ can be represented as the
derivative at~$x$ of some bisection $\beta$ of~$G$ satisfying
$\, \beta(x) = g$, i.e., we can always find $\beta$ such that
$\, u_g^{} = T_x^{} \beta$, but this does of course not mean
that every bisection of~$JG$, as a Lie groupoid over~$M$, can
be written as the jet prolongation of some bisection of~$G$:
those that can be so written are called \emph{holonomous}, and
it is then easy to see that a bisection $\tilde{\beta}$ of $JG$
will be holonomous if and only if $\, \tilde{\beta} = j\beta \,$
where $\,  \beta = \pi_{JG}^{} \smcirc \tilde{\beta}$.
This leads us directly to the definition of the group of
holonomous bisections of~$JG$ and, more generally, of a
full Lie subgroupoid of~$JG$:
\begin{dfn} \label{def:FULLSG}~
 Let\/ $G$ be a Lie groupoid over a manifold\/~$M$.
 We say that a Lie subgroupoid\/~$\tilde{G}$ of its
 jet groupoid\/~$JG$ is \textbf{full} if the projection
 of\/~$JG$ to\/~$G$ remains a bundle projection when
 restricted to\/~$\tilde{G}$.%
 \footnote{It should be noted that this usage of the adjective
 "full", which only makes sense for Lie subgroupoids of a jet
 groupoid, is more restrictive than the one in Ref.~\cite{Mac}.}
\end{dfn}
\begin{dfn} \label{def:HOLBIS}~
 Let $G$ be a Lie groupoid over a manifold $M$ and $\tilde{G}$ a
 full Lie subgroupoid of its jet groupoid $JG$. Then we define the
 \textbf{group of holonomous bisections} of $\tilde{G}$,
 denoted by\/ $H\!B(G,\tilde{G})$, to be the subgroup
 of~$\, \mathrm{Bis}(\tilde{G})$ given by
 \begin{equation}
  H\!B(G,\tilde{G})~
  =~\{ \tilde{\beta} \in \mathrm{Bis}(\tilde{G})~|~\tilde{\beta} = j\beta~
       \mbox{for some $\, \beta \in \mathrm{Bis}(G)$} \} \,,
 \end{equation}
 which can also be viewed as a subgroup of~$\, \mathrm{Bis}(G)$,
 namely, the one given by
 \begin{equation}
  H\!B(G,\tilde{G})~
  \cong~\{ \beta \in \mathrm{Bis}(G)~|~
           j\beta \in \mathrm{Bis}(\tilde{G}) \} \,.
 \end{equation}
\end{dfn}
In what follows, we shall often switch between these two
interpretations without further mention.

The usefulness of this definition stems from the fact that, in general,
the group of holonomous bisections of a full Lie subgroupoid $\tilde{G}$
of~$JG$ cannot be represented as the group of all bisections of some
other Lie groupoid over~$M$. It also leads us to the definition of a
special class of full Lie subgroupoids of a jet groupoid:
\begin{dfn} \label{def:SUBHOL}~
 Let $G$ be a Lie groupoid over a manifold $M$ and $\tilde{G}$ a
 full Lie subgroupoid of its jet groupoid $JG$. Then we say that
 $\tilde{G}$ is \textbf{holonomous} or \textbf{integrable} if it is
 generated by its holonomous bisections, i.e., if for any $\, g \in G$
 with $\, \sigma_G^{}(g) = x \,$ and any $\, u_g^{} \in \tilde{G}_g^{}$,
 there exists $\, \tilde{\beta} \in H\!B(G,\tilde{G}) \,$ such that
 $\, \tilde{\beta}(x) = u_g^{}$. Similarly, we say that $\tilde{G}$ is
 \textbf{locally holonomous} or \textbf{locally integrable} if for any
 $\, g \in G$ with $\, \sigma_G^{}(g) = x \,$ and any $\, u_g^{} \in
 \tilde{G}_g^{}$, there exists an open neighborhood $U$ of~$x$
 and $\, \tilde{\beta} \in H\!B(G|_U,\tilde{G}|_U) \,$ such that
 $\, \tilde{\beta}(x) = u_g^{}$.
\end{dfn}
Let us also observe that the group $H\!B(G,\tilde{G})$ may be
very restricted~-- to the point of collapsing to the trivial group
$\{1\}$~-- but it may also be quite large, depending on the choice
of $G$ and $\tilde{G}$.
(For example, taking $\, \tilde{G} = JG \,$ gives $\, H\!B(G,JG)
\cong \mathrm{Bis}(G)$.)

Here are two important examples, both based on the ``minimal" choice
for~$G$, which is the pair groupoid $M \times M$, so that as stated in
Example~\ref{ex:LFRGR} above, the corresponding jet groupoid $JG$
will be the linear frame groupoid $GL(TM)$ of~$M$.
\begin{exam}~
 Let $M$ be a manifold equipped with a pseudo-riemannian metric $\mathslf{g}$
 and consider its orthonormal frame groupoid $O(TM,\mathslf{g})$ as a Lie
 subgroupoid of the linear frame groupoid $GL(TM)$.
 Then the group of holonomous bisections of $O(TM,\mathslf{g})$ is precisely
 the isometry group of~$(M,\mathslf{g})$ and hence $O(TM,\mathslf{g})$ is
 usually not holonomous: more precisely, it will be (locally) holonomous if and
 only if $(M,\mathslf{g})$ is (locally) strongly isotropic; in particular, it must
 be a space of constant curvature.
\end{exam}
\begin{exam}~
 Let $M$ be a manifold equipped with a symplectic form $\omega$ and
 consider its symplectic frame groupoid $Sp(TM,\omega)$ as a Lie
 subgroupoid of the linear frame groupoid $GL(TM)$.
 Then the group of holonomous bisections of $Sp(TM,\mathslf{g})$ is
 precisely the symplectomorphism group of~$(M,\omega)$ and hence,
 by the Darboux theorem, $Sp(TM,\omega)$ is always locally holonomous.
\end{exam}
Analogous concepts can be introduced for subbundles of jet bundles,
but we do not spell them out explicitly here because they are not needed
for what follows.


In order to handle second order derivatives, we shall need
the \emph{second order jet groupoid} $J^2 G$: as in the case of
fiber bundles, this can be constructed either directly or else by
iteration of the previous construction and subsequent reduction,
which occurs in two steps.
First, observe that the iterated first order jet groupoid $J(JG)$
admits two natural projections to $JG$, namely, the standard
projection $\, \pi_{J(JG)}^{}: J(JG) \longrightarrow JG \,$ and
the jet prolongation $\, J\pi_{JG}^{}: J(JG) \longrightarrow JG$
\linebreak
of the standard projection $\, \pi_{JG}^{}: JG \longrightarrow G$,
explicitly defined as follows: for $\, g \in G$, $u_g^{} \in J_g^{} G \,$
and $\, u'_{u_g} \in  J_{u_g}^{}(JG)$,
\begin{equation} \label{eq:SOJG1}
 (\pi_{J(JG)})_{u_g}^{} (u'_{u_g}) = u_g^{} \,,
\end{equation}
whereas
\begin{equation} \label{eq:SOJG2}
 (J\pi_{JG})_{u_g}^{} (u'_{u_g})
 = T_{u_g}^{} \pi_{JG}^{} \,\smcirc\, u'_{u_g} \,.
\end{equation}
They fit into the following commutative diagram:
\begin{equation} \label{eq:SOJG3}
 \begin{array}{c}
  \xymatrix{
   & J(JG) \ar[dl]_{\pi_{J(JG)}^{}} \ar[dr]^{J\pi_{JG}^{}} & \\
   JG \ar[dr]_{\pi_{JG}^{}} & & JG \ar[dl]^{\pi_{JG}^{}} \\
   & G &
  }
 \end{array}
\end{equation}
The first step of the reduction mentioned above is then again to restrict to
the  subset of $J(JG)$ where the two projections coincide: for $\, g \in G \,$ 
and $\, u_g^{} \in J_g^{} G$, define
\begin{equation} \label{eq:SOJG4}
  \bar{J}_{u_g}^{\,2} G~
  =~\bigl\{ u'_{u_g} \in J_{u_g}^{}(JG) \, | \;
            T_{u_g}^{} \pi_{JG}^{} \,\smcirc\, u'_{u_g} = u_g^{} \bigr\} \,.
\end{equation}
Taking the disjoint union as $u_g^{}$ varies over~$JG$, this defines what
we shall call the \emph{semiholonomic second order jet groupoid} of~$G$,
denoted by $\bar{J}^{\,2} G$: since $\pi_{J(JG)}$ and $J\pi_{JG}$ are both
morphisms of Lie groupoids over $M$, it is naturally a Lie groupoid over~$M$
and it is also a fiber bundle over~$JG$.
The second step consists in decomposing this, as a fiber product of fiber
bundles over~$JG$, into a symmetric part and an antisymmetric part:
the former is precisely $J^2 G$ whereas the latter does not seem to
play any significant role in the theory; we shall therefore not pursue
this decomposition any further. But what will be important in the sequel
is that, once again, given any bisection $\tilde{\beta}$ of $JG$, its jet
prolongation $j\tilde{\beta}$ will be a bisection of $J(JG)$ which will take
values in $\bar{J}^{\,2} G$ if and only if $\tilde{\beta}$ is holonomous,
in which case $\, \tilde{\beta} = j\beta \,$ and $\, j\tilde{\beta} =
j^2 \beta \,$ where $\, \beta = \pi_{JG}^{} \smcirc \tilde{\beta}$.
Using the terminology introduced above, we may summarize this
in the statement that the group of holonomous bisections of~%
$\bar{J}^{\,2} G$ is precisely the group of bisections of~$J^2 G$,
which in turn is isomorphic to the group of bisections of~$G$:
\begin{equation} \label{eq:HOLBIS}
 H\!B(JG,\bar{J}^{\,2} G)~
 =~\mathrm{Bis}(J^2 G)~\cong~\mathrm{Bis}(G) \,.
\end{equation}
We may reexpress this as the statement that $J^2 G$ is the unique
maximal holonomous Lie subgroupoid of~$\bar{J}^{\,2} G$.

\subsection{Lie algebroids}

Just as Lie algebras are the infinitesimal version of Lie groups,
Lie algebroids are the infinitesimal version of Lie groupoids.
Formally, a Lie algebroid \emph{over} a manifold~$M$ is a
vector bundle $\mathfrak{g}$ over~$M$ endowed with two
structure maps which, when dealing with various Lie algebroids
at the same time, we shall decorate with an index~$\mathfrak{g}$,
namely, the \emph{anchor} $\, \alpha_{\mathfrak{g}}: \mathfrak{g}
\longrightarrow TM$, which is required to be a homomorphism of
vector bundles over~$M$ and, by push-forward of sections, induces
a homomorphism $\, \alpha_{\mathfrak{g}}: \varGamma(\mathfrak{g})
\longrightarrow \mathfrak{X}(M) \,$ of modules over the function ring
$C^\infty(M)$ (we follow the common abuse of notation of denoting
both by the same symbol), and the \emph{bracket}
\begin{equation} \label{eq:BRACK}
 [.\,,.]_{\mathfrak{g}}:~
 \varGamma(\mathfrak{g}) \times \varGamma(\mathfrak{g})~~
 \longrightarrow~~\varGamma(\mathfrak{g}) \,,
\end{equation}
satisfying the usual axioms of Lie algebra theory (bilinearity over
$\mathbb{R}$, antisymmetry and the Jacobi identity), together with
the Leibniz identity
\begin{equation} \label{eq:LEIBN}
 [X,fY]_{\mathfrak{g}}~
 = \; f[X,Y]_{\mathfrak{g}} \, + \, (L_{\alpha_{\mathfrak{g}}(X)} f) \, Y
 \qquad \mbox{for $\, f \in C^\infty(M)$,
                      $X,Y \in \varGamma(\mathfrak{g})$} \,.
\end{equation}
Once again, we gather some standard terminology.
For example, $\mathfrak{g}$ is said to be \emph{transitive} if the
anchor is surjective and \emph{totally intransitive} if it is zero.
More generally, $\mathfrak{g}$ is said to be \emph{regular} if
the anchor has constant rank.
In this case, the kernel of the anchor,
\begin{equation} \label{eq:KANCH}
 \mathfrak{g}_{\mathrm{iso}}~=~\ker \alpha_{\mathfrak{g}} \,,
\end{equation}
is a totally intransitive Lie subalgebroid of $\mathfrak{g}$, called
its \emph{isotropy algebroid}, and when it is locally trivial (which
may not always be the case), it will in fact be a \emph{Lie algebra
bundle} over~$M$.
Similarly, in this case, the image of the anchor,
\begin{equation} \label{eq:IANCH}
 \mathfrak{X}^{\mathfrak{g}}(M)~
 =~\mathrm{im} \, \alpha_{\mathfrak{g}} \,,
\end{equation}
is an involutive distribution on~$M$ called its \emph{characteristic
distribution}.

\pagebreak

We shall skip the definition of morphism between Lie algebroids (over
the same manifold), which is the obvious one, as well as the discussion
of examples, and just mention that the insertion of the Lie algebroid
concept into Lie algebra theory is even more direct than in the group
case, since the space $\varGamma(\mathfrak{g})$ of sections of a Lie
algebroid $\mathfrak{g}$ is a Lie algebra by definition.
Moreover, at least in the regular case (which is the only one of interest
here), $\varGamma(\mathfrak{g})$ has a natural Lie subalgebra, namely
the Lie algebra $\varGamma(\mathfrak{g}_{\mathrm{iso}}^{})$ of sections
of the isotropy algebroid of~$\mathfrak{g}$, and a natural quotient Lie
algebra, namely the Lie algebra $\mathfrak{X}^{\mathfrak{g}}(M)$ of
sections of the characteristic distribution on~$M$, which fit together
into the following exact sequence of Lie algebras:
\begin{equation}
 \{0\}~\longrightarrow~\varGamma(\mathfrak{g}_{\mathrm{iso}}^{})
     ~\longrightarrow~\varGamma(\mathfrak{g})~\longrightarrow
     ~\mathfrak{X}^{\mathfrak{g}}(M)~\longrightarrow~\{0\} \,.
\end{equation}

Next, let us pass to the construction of the Lie algebroid of a
Lie groupoid, which once again follows closely that of the Lie
algebra of a Lie group; the same goes for the exponential.
\linebreak
Given a Lie groupoid $G$ over a manifold~$M$, we use the unit map
$1_G^{}$ to consider $M$ as an embedded submanifold of~$G$ and
restrict the tangent maps to the source and target projections to this
submanifold, which provides us with two homomorphisms of vector
bundles over~$M$ that we can put into the following diagram:
\begin{equation}
 \begin{array}{c}
  \xymatrix{
   TG|_M^{} \ar[rd]
   \ar@<0.3ex>[rr]^{T\sigma_G|_M}
   \ar@<-0.3ex>[rr]_{T\tau_G|_M}
   &   & TM \ar[ld] \\
   & M &
  }
 \end{array}
\end{equation}
As a vector bundle over~$M$, the corresponding Lie algebroid~%
$\mathfrak{g}$ is then defined as the restriction to~$M$ of the
vertical bundle with respect to the source projection:  
\begin{equation} \label{eq:ALL1}
 \mathfrak{g}~=~(V^\sigma G)|_M~
 =~\ker T\sigma_G^{}|_M^{} \,.
\end{equation}
More explicitly, this means that, for any $\, x \in M$, $\mathfrak{g}_x$
is the tangent space to the source fiber at~$x$:
\begin{equation} \label{eq:ALL2}
 \mathfrak{g}_x^{}~=~V_x^\sigma G~=~T_x^{}(G_x^{}) \,.
\end{equation}
The anchor of $\mathfrak{g}$ is then defined as the restriction
of the target projection:
\begin{equation} \label{eq:ALL3}
 \alpha_{\mathfrak{g}}~
 =~(T\tau_G^{}|_M^{}) \big|_{\mathfrak{g}}^{} \,.
\end{equation}
To construct the bracket between sections of $\mathfrak{g}$, we
note that given $g$ in~$G$ with $\, \sigma_G^{}(g) = x \,$ and
$\, \tau_G^{}(g) = y$, right translation by $g$ is a diffeomorphism
from the source fiber at~$y$ to the source fiber at~$x$:
\begin{equation}
 \begin{array}{cccc}
  R_g: & G_y & \longrightarrow & G_x \\[0.5ex]
       &  h  &   \longmapsto   & hg
 \end{array}
\end{equation}
Thus its derivative at the point $h \in G_y$ is a linear isomorphism
$\, T_h R_g: T_h (G_y) \longrightarrow T_{hg} (G_x)$, and hence
given a vector field $Z$ on~$G$ that is $\sigma_G^{}$-vertical,
$Z \in \varGamma(V^\sigma G)$, it makes sense to say that $Z$
is \emph{right invariant} if, for any $\, (h,g) \in G \times_M G$,
\[
 T_h R_g (Z(h))~=~Z(hg) \,.
\]
Now observe that (a) the space $\mathfrak{X}^{ri}(G)$ of right
invariant vector fields on~$G$ is a Lie subalgebra of the space
$\mathfrak{X}(G)$ of all vector fields on~$G$ and (b) restriction
to~$M$ provides a linear isomorphism from $\mathfrak{X}^{ri}(G)$ to
the space $\varGamma(\mathfrak{g})$ of sections of~$\mathfrak{g}$
whose inverse can be described explicitly as follows: given $\, X \in
\varGamma(\mathfrak{g})$, the corresponding right invariant vector
field $\, X^r \in \mathfrak{X}^{ri}(G)$ is defined by
\[
 X^r(g)~=~T_{\tau_G^{}(g)} R_g^{}  \bigl( X(\tau_G^{}(g)) \bigr) \,.
\]
This linear isomorphism is used to transfer the structure of Lie algebra
from~$\mathfrak{X}^{ri}(G)$ to $\varGamma(\mathfrak{g})$. \linebreak
The same idea is used to construct the \emph{exponential}, as a map
\begin{equation} \label{eq:EXP1}
 \exp:~\varGamma(\mathfrak{g})~\longrightarrow~\mathrm{Bis}(G)
\end{equation}
defined by taking the flow $F_{X^r}$ of the right invariant vector field
$X^r$ on~$G$ corresponding to $X$ at time $1$:
\[
 \exp(X)(x)~=~F_{X^r}(1,1_x) \qquad \mbox{for $\, x \in M$} \,.
\]
Then, more generally,
\begin{equation} \label{eq:EXP2}
 F_{X^r}(t,g)~=~\exp(tX)(y) \, g
 \qquad \mbox{for $\, t \in \mathbb{R}, \, g \in {}_y G_x$} \,,
\end{equation}
and conversely,
\begin{equation} \label{eq:EXP3}
 \exp(tX)(x) \big|_{t=0}~=~1_x~~,~~
 \frac{d}{dt} \exp(tX)(x) \Big|_{t=0}~=~X(x)
 \qquad \mbox{for $\, x \in M$}.
\end{equation}

Continuing the analogies between Lie group theory and Lie groupoid
theory, there is also the concept of an \emph{infinitesimal action}
of a Lie algebroid $\mathfrak{g}$ on a fiber bundle~$E$, both over
the same base manifold~$M$: denoting by $\mathfrak{X}^P(E)$
the Lie algebra of projectable vector fields on~$E$,%
\footnote{A vector field $X_E$ on the total space~$E$ of a fiber bundle
is said to be projectable if for any two points $e_1,e_2$ in the same fiber,
we have $T_{e_1} \pi_E (X_E(e_1)) = T_{e_2} \pi_E (X_E(e_2))$, that is,
if there exists a (necessarily unique) vector field $X_M$ on the base
space~$M$ to which it is $\pi_E$-related: one then says that $X_E$
projects to~$X_M$. Vertical vector fields are those that project to $0$.}
this is defined to be a linear map
\begin{equation} \label{eq:ACTLA1}
 \begin{array}{cccc}
  \dot{\Phi}_E:
  & \varGamma(\mathfrak{g}) & \longrightarrow & \mathfrak{X}^P(E) \\[1mm]
  &            X            &   \longmapsto   &        X_E
  \end{array}
\end{equation}
which is compatible with the structures involved: it is linear over the
pertinent function rings, i.e.,
\begin{equation} \label{eq:ACTLA2}
 (fX)_E^{}~=~(f \smcirc \pi_E^{}) \, X_E^{}
 \qquad \mbox{for $\, f \in C^\infty(M)$, $X \in \varGamma(\mathfrak{g})$} \,,
\end{equation}
takes the anchor to the projection, i.e.,
\begin{equation} \label{eq:ACTLA3}
 X_M^{}~=~\alpha_{\mathfrak{g}}(X)
 \qquad \mbox{for $\, X \in \varGamma(\mathfrak{g})$} \,,
\end{equation}
and preserves brackets, i.e., is a homomorphism of Lie algebras; we
call $X_E^{}$ the \emph{fundamental vector field} associated to $X$. 
The terminology is justified by noting that if $\mathfrak{g}$ is the Lie
algebroid of a Lie groupoid~$G$, then an action of $G$ on $E$ induces
an infinitesimal action of $\mathfrak{g}$ on $E$ defined as follows:
given $\, X \in \varGamma(\mathfrak{g})$ and recalling that, for any
$\, x \in M$,
\begin{equation} \label{eq:ACTLA4}
 X(x)~=~\frac{d}{dt} \, \exp(tX)(x) \, \Big|_{t=0},
\end{equation}
we have, for any $\, e \in E \,$ with $\, x = \pi_E^{}(e) \in M$,
\begin{equation} \label{eq:ACTLA5}
 X_E^{}(e)~=~\frac{d}{dt} \, \exp(tX)(x) \cdot e \, \Big|_{t=0}.
\end{equation}
In fact, the above definition of infinitesimal action provides nothing
more and nothing less than a representation of the Lie algebra of sections
of the Lie algebroid $\mathfrak{g}$ by projectable vector fields on~$E$
which is the formal derivative of the representation $\Pi_E^{}$ of the
group of bisections of the Lie groupoid $G$ by automorphisms of~$E$
in equations~(\ref{eq:ACTREP1}) and~(\ref{eq:ACTREP2}).
Note that the fundamental vector field $X_E^{}$ associated to a
section $X$ of the isotropy algebroid of~$\mathfrak{g}$ will be
vertical and hence we get the following commutative diagram,
\begin{equation} \label{eq:EXSEQ2}
 \begin{array}{c}
  \xymatrix{
   \{0\} \ar[r] & \mathfrak{X}^V(E) \ar[r] &
   \mathfrak{X}^P(E) \ar[r] & \mathfrak{X}^E(M) \ar[r] & \{0\} \\
   \{0\} \ar[r] & \varGamma(\mathfrak{g}_{\mathrm{iso}}) \ar[r] \ar[u] &
   \varGamma(\mathfrak{g}) \ar[r] \ar[u] &
   \mathfrak{X}^{\mathfrak{g}}(M) \ar[r] \ar[u] & \{0\}
  }
 \end{array}
\end{equation}
where $\mathfrak{X}^V(E)$ denotes the Lie algebra of vertical vector
fields on~$E$ and $\mathfrak{X}^E(M)$ denotes the algebra of vector
fields on $M$ admitting some lift to a projectable vector field on~$E$. 

\subsection{The jet algebroid of a Lie algebroid}

As in the case of Lie groupoids, there are various ways of constructing
new Lie algebroids from a given one.
Here, we want to mention just one of them, namely that of the jet
algebroid of a Lie algebroid.
To do so, we use the functor $J$ on fiber bundles and observe that,
for any manifold~$M$, the jet bundle of a fiber bundle over~$M$
carrying some additional structure will in many cases inherit that
additional structure.
In particular, if $E$ is a vector bundle over~$M$, its jet bundle $JE$
is again a vector bundle over~$M$: this can be most easily seen by
writing points in~$JE$ as values of jets of sections of~$E$ and defining
a linear combination of points in~$JE$ over the same point of~$M$ as
the value of the jet of the corresponding linear combination of sections
of~$E$: this is the necessary and sufficient condition for the jet
prolongation of sections (from sections of~$E$ to sections of~$JE$)
to be a linear map $\, j: \varGamma(E) \longrightarrow \varGamma(JE)$.
\begin{prp}~
 Given a Lie algebroid $\mathfrak{g}$ over a manifold $M$ with anchor
 $\alpha_{\mathfrak{g}}^{}$ and bracket $[.\,,.]_{\mathfrak{g}}^{}$,
 its jet bundle $J\mathfrak{g}$ is again a Lie algebroid over~$M$,
 with anchor $\alpha_{J\mathfrak{g}}^{}$ and bracket $[.\,,.]%
 _{J\mathfrak{g}}^{}$ defined so that the jet prolongation of sections
 (from sections of~$\mathfrak{g}$ to sections of~$J\mathfrak{g}$)
 preserves the anchor and the bracket:
 \begin{equation} \label{eq:JETALG1}
  \alpha_{J\mathfrak{g}}^{}(j\xi)~=~\alpha_{\mathfrak{g}}^{}(\xi)~~,~~
  [j\xi,j\eta]_{J\mathfrak{g}}^{}~=~j([\xi,\eta]_{\mathfrak{g}})^{} 
  \qquad \mbox{for $\, \xi,\eta \in \varGamma(\mathfrak{g})$} \,.
\end{equation}
\end{prp}
We note that this condition of compatibility with the procedure of
taking  the jet prolongation of sections fixes the anchor and the
bracket in $J\mathfrak{g}$ uniquely.
As a vector bundle over~$M$,
\begin{equation} \label{eq:JETALG2}
 J\mathfrak{g}~=~\mathfrak{g} \oplus L(TM,\mathfrak{g})~
 =~\mathfrak{g} \oplus \bigl( T^* M \otimes \mathfrak{g} \bigr) \,,
\end{equation}
and this direct decomposition is also adequate for describing
the anchor ($\alpha_{J\mathfrak{g}}^{}$ is just projection onto
the first summand followed by $\alpha_{\mathfrak{g}}^{}$) but
not for describing the bracket.
To do that, one can employ various methods, but the most
transparent one of them, at least for physicists, is in terms of the
corresponding structure functions, using the fact that the anchors
and brackets, in $\mathfrak{g}$ as well as in $J\mathfrak{g}$,
are differential operators in each argument, and hence local.
(More precisely, the anchors are differential operators of
order~$0$, since they result from push-forward of sections
by a vector bundle homomorphism, while the brackets are
bidifferential operators of order~$1$, since they satisfy a
Leibniz rule, as in equation~(\ref{eq:LEIBN}), in each factor.)
Namely, given local coordinates $x^\mu$ of~$M$ and a basis
of local sections $T_a$ of~$\mathfrak{g}$, together with
the induced basis of local sections $(T_a,T_a^\mu)$
of~$J\mathfrak{g}$ ($T_a^\mu = dx^\mu \otimes T_a$),
we have that if $\xi$ is a section of~$\mathfrak{g}$, locally
represented as
\begin{equation} \label{eq:JETALG3}
  \xi = \xi^a T_a^{} \,,
\end{equation}
then $j\xi$ is a section of~$J\mathfrak{g}$ locally represented as
\begin{equation} \label{JETALG4}
  j\xi = \xi^a T_a^{} + \partial_\mu \xi^a \, T_a^\mu \,.
\end{equation}
Assuming that the anchor and the bracket of $\mathfrak{g}$ can
be written in terms of structure functions $f_{a\vphantom{i}}^\mu$
and~$f_{ab}^c$, according to
\begin{equation} \label{eq:JETALG5}
 \alpha_{\mathfrak{g}}^{}(T_a^{})~
 =~f_{a\vphantom{i}}^\mu \, \partial_\mu^{} \,,
\end{equation}
 and
\begin{equation} \label{eq:JETALG6}
 [T_a^{},T_b^{}]_{\mathfrak{g}}^{}~=~f_{ab}^c T_c^{} \,,
\end{equation}
the anchor and the bracket of~$J\mathfrak{g}$ are given by
\begin{equation} \label{eq:JETALG7}
 \begin{array}{c}
  \alpha_{J\mathfrak{g}}^{}(T_a^{})~
  =~f_{a\vphantom{i}}^\mu \, \partial_\mu^{} \\[2mm]
  \alpha_{J\mathfrak{g}}^{}(T_{a\vphantom{i}}^\mu)~=~0
 \end{array} \,,
\end{equation}
and
\begin{equation} \label{eq:JETALG8}
 \begin{array}{c}
  [T_a^{},T_b^{}]_{J\mathfrak{g}}^{}~
  =~f_{ab}^{c\vphantom{\mu}} T_c^{} \, + \, \partial_\mu^{}
    f_{ab}^{c\vphantom{\mu}} \, T_{c\vphantom{i}}^\mu \\[2mm]
  [T_a^{},T_b^\mu]_{J\mathfrak{g}}^{}~
  =~f_{ab}^{c\vphantom{\mu}} T_{c\vphantom{i}}^\mu \, + \,
    \partial_\nu^{} f_{a\vphantom{i}}^\mu \, T_b^{\nu\vphantom{p}} \\[2mm]
  [T_{a\vphantom{b}}^\mu,T_b^{\nu\vphantom{\mu}}]_{J\mathfrak{g}}^{}~
  =~f_{a\vphantom{i}}^{\nu\vphantom{\mu}} T_b^\mu \, - \,
    f_b^{\nu\vphantom{\mu}} T_{a\vphantom{i}}^\mu
 \end{array} \,.
\end{equation}
We leave it to the reader to convince himself that, up to a canonical
isomorphism, the functor $J$ commutes with the process of passing
from a Lie groupoid to its corresponding Lie algebroid, that is, if
$\mathfrak{g}$ is the Lie algebroid associated to the Lie groupoid
$G$, then $J\mathfrak{g}$ is the Lie algebroid associated to the
Lie groupoid~$JG$.

\section{Induced actions}

Starting from a given action of a Lie groupoid on a fiber bundle,
we shall in this section show how to construct ``induced actions'' of
certain other Lie groupoids, derived from the original one, on certain
other fiber bundles, derived from the original one: this is an essential
technical feature needed to make the theory work.
Throughout the section, we maintain the notation used before: $G$
will be a Lie groupoid over a manifold~$M$, with source projection
$\sigma_G^{}$ and target projection $\tau_G^{}$, multiplication
$\mu_G^{}$, unit map~$1_G^{}$ and inversion~$\iota_G^{}$,
and $E$ will be a fiber bundle over the same manifold~$M$, with
projection $\pi_E^{}$, endowed with an action $\Phi_E^{}$ as in
equation~(\ref{eq:ACTLG1}).
Note that under these circumstances, a groupoid element
$g \in G$ with source $\, \sigma_G^{}(g) = x \,$ and target
$\, \tau_G^{}(g) = y \,$ will provide a diffeomorphism
$\, L_g: E_x \longrightarrow E_y \,$ called
\emph{left translation} by~$g$.

Our first and most elementary example of such an induced action is that
of the original Lie groupoid $G$ on the vertical bundle $VE$ of~$E$,
\begin{equation} \label{eq:IACT01}
 \begin{array}{cccc}
  \Phi_{VE}^{}: & G \times_M VE & \longrightarrow &      VE        \\[1mm]
                &  (g,v_e^{})   &   \longmapsto   & g \cdot v_e^{}
 \end{array}
\end{equation}
defined by requiring left translation by~$g$ in~$VE$ to be simply the
derivative of left translation by~$g$ in~$E$, i.e., given $g \in G$ with
$\, \sigma_G^{}(g) = x \,$ and $\, \tau_G^{}(g) = y$, $e \in E \,$ with
$\, \pi_E^{}(e) = x \,$ and a vertical vector $\, v_e^{} \in V_e^{} E$,
we have
\begin{equation} \label{eq:IACT02}
 g \cdot v_e^{}~=~T_e^{} L_g^{} (v_e^{}) \,.
\end{equation}
This means that given a vertical curve $\, t \mapsto e(t)$ in $E$
passing through the point $e$ (i.e., $e(0) = e$), we have
\begin{equation} \label{eq:IACT03}
 g \cdot \frac{d}{dt} \, e(t) \Big|_{t=0}~
 =~\frac{d}{dt} \; g \cdot e(t) \Big|_{t=0} \,.
\end{equation}
An important property of the action~(\ref{eq:IACT01}) is that it respects 
the structure of $VE$ as a vector bundle over $E$, since (a) it covers the
original action, i.e., the diagram
\begin{equation} \label{eq:IACT04}
 \begin{array}{c}
  \xymatrix{
   G \times_M VE \ar[rr]^{\qquad \Phi_{VE}^{}} \ar[d] & & VE \ar[d] \\
   G \times_M E \ar[rr]_{\qquad \Phi_E^{}} & & E
  }
 \end{array}
\end{equation}
commutes, and (b)~for any $\, g \in G$, left translation by~$g$
is well defined on the whole vertical space $V_e^{} E$ provided
that $\, \pi_E^{}(e) = \sigma_G^{}(g) \,$ (and otherwise is not well
defined at any point in $V_e^{} E$), its restriction to this space
being a linear transformation $\, L_g^{}: V_e^{} E \longrightarrow
V_{g \cdot e}^{} E$, since it is the tangent map at~$e$ to left
translation $\, L_g^{}: E_x^{} \longrightarrow E_y^{}$.

Our next example of an induced action is functorial.
Namely, applying the jet functor to all structural maps of the original
action, we obtain an action of the jet groupoid $JG$ of~$G$ on the jet
bundle $JE$ of~$E$,
\begin{equation} \label{eq:IACT05}
 \begin{array}{cccc}
  \Phi_{JE}: & JG \times_M JE & \longrightarrow &      JE       \\[1mm]
             &   (u_g,u_e)    &   \longmapsto   & u_g \cdot u_e
 \end{array}
\end{equation}
defined as follows: given $\, g \in G \,$ with $\, \sigma_G^{}(g) = x \,$
and $\, \tau_G^{}(g) = y$, $e \in E \,$ with $\, \pi_E^{}(e) = x \,$ and
jets $\, u_g^{} \in J_g^{} G \subset L(T_x^{} M,T_g^{} G) \,$ and
$\, u_e^{} \in J_e^{} E \subset L(T_x^{} M,T_e^{} E)$, we concatenate
both in a linear map $\, (u_g^{},u_e^{}) \in L(T_x^{} M , T_g^{} G
\oplus T_e^{} E) \,$ (which actually takes values in $\, L(T_x^{} M,
T_{(g,e)}(G \times_M E)) \,$ since $\, T_x^{} \sigma_G^{}
\smcirc u_g^{} = \mathrm{id}_{T_x M} = T_e^{} \pi_E^{}
\smcirc u_e^{}$) and compose:
\begin{equation} \label{eq:IACT06}
 u_g^{} \cdot u_e^{}~
 =~T_{(g,e)}^{} \Phi_E^{} \smcirc (u_g^{},u_e^{}) \smcirc
   (T_g^{} \tau_G^{} \smcirc u_g^{})^{-1} \,.
\end{equation}
Essentially the same procedure also provides an action of the jet
groupoid $JG$ of~$G$ on the linearized jet bundle $\vec{J} E$ of~$E$,
\begin{equation} \label{eq:IACT07}
 \begin{array}{cccc}
  \Phi_{\vec{J} E}:
  & JG \times_M \vec{J} E & \longrightarrow &      \vec{J} E      \\[1mm]
  &    (u_g,\vec{u}_e)    &   \longmapsto   & u_g \cdot \vec{u}_e
 \end{array}
\end{equation}
Unlike the previous one, this admits a simplification because it factorizes
through the morphism~(\ref{eq:JETGP3}) of Lie groupoids to yield an action
of the Lie groupoid $\, GL(TM) \times_M G \,$ on the linearized jet bundle
$\vec{J} E$ of~$E$,
\begin{equation} \label{eq:IACT08}
 \begin{array}{ccc}
  \bigl( GL(TM) \times_M G \bigr) \times_M \vec{J} E
  & \longrightarrow & \vec{J} E \\[1mm]
             \bigl( (a,g),\vec{u}_e \bigr)          
  &   \longmapsto   & (a,g) \cdot \vec{u}_e
 \end{array}
\end{equation}
This action is suggested by the isomorphism
\begin{equation} \label{eq:LJB3}
 \vec{J} E~\cong~\pi_E^* \bigl( T^* M \bigr) \otimes VE
\end{equation}
(see equation~(\ref{eq:LJB2})), together with the fact that the tangent
bundle and the cotangent bundle of~$M$ are endowed with natural actions
of the linear frame groupoid $GL(TM)$ and the vertical bundle of~$E$ is
endowed with the induced action of~$G$ as explained above.
However, it should be noted that all the groupoids involved are groupoids
over~$M$ while the isomorphism~(\ref{eq:LJB3}) is one of vector bundles
over~$E$.
Therefore, it is worthwhile specifying that the action~(\ref{eq:IACT08})
is explicitly defined as follows: given $\, (a,g) \in GL(TM) \times_M G \,$
with $\, \sigma_{GL(TM)}(a) = x = \sigma_G(g) \,$ and $\, \tau_{GL(TM)}(a)
= y = \tau_G(g)$, $e \in E \,$ with $\, \pi_E(e) = x \,$ and $\, \vec{u}_e
\in \vec{J}_e E = L(T_x M,V_e E)$, we obtain $\, (a,g) \cdot \vec{u}_e \in
\vec{J}_{g \cdot e} E = L(T_y M,V_{g \cdot e} E) \,$ by composition: 
\begin{equation} \label{eq:IACT09}
 (a,g) \cdot \vec{u}_e~
 =~T_e L_g \,\smcirc\, \vec{u}_e \,\smcirc\, a^{-1} \,.
\end{equation}

An important property of the actions~(\ref{eq:IACT05}) and~%
(\ref{eq:IACT07}) is that they respect the structure of~$JE$ as
an affine bundle and of $\vec{J} E$ as a vector bundle over $E$, 
since (a) they cover the original action, i.e., the diagrams
\begin{equation} \label{eq:IACT10}
 \begin{array}{c}
  \xymatrix{
   JG \times_M JE \ar[rr]^{\qquad \Phi_{JE}^{}} \ar[d] & & JE \ar[d] \\
   G \times_M E \ar[rr]_{\qquad \Phi_E^{}} & & E
  }
 \end{array}
\end{equation}
and
\begin{equation} \label{eq:IACT11}
 \begin{array}{c}
  \xymatrix{
   JG \times_M \vec{J} E \ar[rr]^{\qquad \Phi_{\vec{J} E}} \ar[d] & &
   \vec{J} E \ar[d] \\
   G \times_M E \ar[rr]_{\qquad \Phi_E^{}} & & E
  }
 \end{array}
\end{equation}
commute, and (b)~for any $\, g \in G$ and $u_g \in J_g G$, left translation 
by~$u_g$ is well defined on the whole jet space $J_e E$ and on the whole
linearized jet space $\vec{J}_e E$ provided that $\, \pi_E(e) = \sigma_G(g)$
\linebreak
(and otherwise is not well defined at point in $J_e E$ or $\vec{J}_e E$), its
restriction to these spaces being an affine transformation $\, L_{u_g}:
J_e^{} E \longrightarrow J_{g \cdot e} E$ or a linear transformation
$\, L_{u_g}: \vec{J}_e^{} E \longrightarrow \vec{J}_{g \cdot e} E$.

The compatibility of the actions~(\ref{eq:IACT05}) and~(\ref{eq:IACT07})
with the structure of $JE$ as an affine bundle and of $\vec{J} E$ as a
vector bundle over~$E$ allows us to transfer these actions from jets to
cojets (ordinary or twisted), by dualization.
Concentrating on the twisted case, which is the more important one for
the applications we have in mind, we obtain an action
\begin{equation} \label{eq:IACT12}
 \begin{array}{cccc}
  \Phi_{J^{\circledstar} E}:
  & JG \times_M J^{\circledstar} E & \longrightarrow & J^{\circledstar} E
  \\[1mm]
  &            (u_g,z_e)           &   \longmapsto   &   u_g \cdot z_e
 \end{array}
\end{equation}
of $JG$ on $J^{\circledstar} E$ defined as follows: given $\, g \in G \,$
with $\, \sigma_G^{}(g) = x \,$ and $\, \tau_G^{}(g) = y$, $e \in E \,$
with $\, \pi_E^{}(e) = x$, $u_g \in J_g G$, $z_e \in J^{\circledstar}_e E \,$
and $\, u_{g \cdot e} \in J_{g \cdot e} E$,
\begin{equation} \label{eq:IACT13}
 \langle u_g \cdot z_e^{} \,,\, u_{g \cdot e} \rangle~
 =~(T_g \tau_G^{} \smcirc u_g)^{-1^{\,\scriptstyle{*}}}
   \langle z_e^{} \,,\, u_g^{-1} \cdot u_{g \cdot e}^{} \rangle \,.
\end{equation}
In other words, we require the following diagram to commute:
\begin{equation} \label{eq:IACT14}
 \begin{array}{c}
  \xymatrix{
   J_e E~ \ar[rr]^{L_{u_g}} \ar[d]_{z_e} & &
   ~J_{g \cdot e} E \ar[d]^{u_g \cdot z_e} \\
   \bwedge^{\!n\,} T_x^* M~ \ar[rr]_{(T_g \tau_G^{} \smcirc u_g)^{-1^{\,*}}} & &
   ~\bwedge^{\!n\,} T_y^* M
  }
 \end{array}
\end{equation}
Similarly, we obtain an action
\begin{equation} \label{eq:IACT15}
 \begin{array}{cccc}
  \Phi_{\vec{J}^{\,\circledast} E}:
  & JG \times_M \vec{J}^{\,\circledast} E & \longrightarrow &
  \vec{J}^{\,\circledast} E \\[1mm]
  &           (u_g,\vec{z}_e)             &   \longmapsto   &
  u_g \cdot \vec{z}_e
 \end{array}
\end{equation}
of $JG$ on $\vec{J}^{\,\circledast} E$ that factorizes through the
composition of the morphism~(\ref{eq:JETGP3}) of Lie groupoids to
yield an action
\begin{equation} \label{eq:IACT16}
 \begin{array}{ccc}
  \bigl( GL(TM) \times_M G \bigr) \times_M \vec{J}^{\,\circledast} E
  & \longrightarrow & \vec{J}^{\,\circledast} E \\[1mm]
                   \bigl( (a,g),\vec{z}_e \bigr)                  
  &   \longmapsto   & (a,g) \cdot \vec{z}_e
 \end{array}
\end{equation}
of $\, GL(TM) \times_M G \,$ on $\vec{J}^{\,\circledast} E$
defined as follows: given $\, (a,g) \in GL(TM) \times_M G \,$
with $\, \sigma_{GL(TM)}(a) = x = \sigma_G(g) \,$  and
$\, \tau_{GL(TM)}(a) = y = \tau_G(g)$, $e \in E \,$ with
$\, \pi_E(e) = x$, $\vec{z}_e \in \vec{J}^{\,\circledast}_e E \,$
and $\, \vec{u}_{g \cdot e} \in \vec{J}_{g \cdot e} E$,
\begin{equation} \label{eq:IACT17}
 \langle (a,g) \cdot \vec{z}_e \,,\, \vec{u}_{g \cdot e} \rangle~
 =~a^{-1^{\,\scriptstyle{*}}}
   \langle \vec{z}_e \,,\, (a,g)^{-1} \cdot \vec{u}_{g \cdot e} \rangle \,.
\end{equation}
In other words, we require the following diagram to commute:
\begin{equation} \label{eq:IACT18}
 \begin{array}{c}
  \xymatrix{
   \vec{J}_e E \ar[rr]^{L_{(a,g)}} \ar[d]_{\vec{z}_e} & &
   \vec{J}_{g \cdot e} E \ar[d]^{(a,g) \cdot \vec{z}_e} \\
   \bwedge^{\!n\,} T_x^* M \, \ar[rr]_{a^{-1^{\,*}}} & &
   \bwedge^{\!n\,} T_y^* M
  }
 \end{array}
\end{equation}
All these actions again satisfy the property of compatibility with
the structure of the bundles involved as vector bundles over $E$.

Passing to our next example, which is once again functorial, let us
now apply the tangent functor $T$ to all structural maps of the
original action to obtain an action of the tangent groupoid $TG$
of~$G$ on the tangent bundle $TE$ of~$E$, 
\begin{equation} \label{eq:IACT19}
 \begin{array}{ccc}
  TG \times_{TM} TE & \longrightarrow &      TE       \\[1mm]
      (v_g,v_e)     &   \longmapsto   & v_g \cdot v_e
 \end{array}
\end{equation}
where we have used the canonical identification of
$\, T \bigl( G \times_M E \bigr)$ with $\, TG \times_{TM} TE$, under
which a pair of vectors $\, (v_g,v_e) \in T_g G \oplus  T_e E \,$
belongs to the subspace $\, T_{(g,e)} \bigl( G \times_M E \bigr) =
(TG \times_{TM} TE)_{(g,e)}$ \linebreak if and only if its two
components are related according to $\, T_g \sigma_G (v_g)
= T_e \pi_E (v_e)$, and in this case, \vspace{-1ex}
\begin{equation} \label{eq:IACT20}
 v_g \cdot v_e~=~T_{(g,e)} \Phi_E^{} (v_g,v_e) \,.
\vspace{0.5ex}
\end{equation}
Now even though this action still covers the original one, i.e., the diagram
\begin{equation} \label{eq:IACT21}
 \begin{array}{c}
  \xymatrix{
   TG \times_{TM} TE \ar[rr]^{\qquad T\Phi_E^{}} \ar[d] & & TE \ar[d] \\
   G \times_M E \ar[rr]_{\qquad \Phi_E^{}} & & E
  }
 \end{array}
\end{equation}
commutes, the problem is that it involves a change of base space, from~$M$
to~$TM$, and as a result it does not respect the structure of $TE$ as a vector
bundle over $E$.
Namely, given $\, g \in G \,$ with $\, \sigma_G(g) = x \,$ and $\, \tau_G(g)
= y$, $e \in E \,$ with $\, \pi_E(e) = x \,$ and $\, v_g \in T_g G \,$ with
$\, T_g \sigma_G(v_g) = v_x \in T_x M \,$ and $\, T_g \tau_G(v_g) = v_y 
\in T_y M$, left translation by~$v_g$ is not well defined on the
whole tangent space $T_e E$, but only on its affine subspace
$(T_e \pi_E)^{-1}(v_x)$, and its restriction to this subspace is an
affine transformation $\, L_{v_g}: (T_e \pi_E)^{-1}(v_x) \longrightarrow
(T_{g \cdot e} \pi_E)^{-1}(v_y)$.

This is a serious defect because it prevents the transfer of this action
to cotangent vectors or, more generally, tensors on~$E$ and thus makes it
almost useless.

Fortunately, and that is perhaps the central message of this paper, there is
a way out of this impasse: it consists in replacing the tangent groupoid $TG$
by the jet groupoid~$JG$.
In fact, as we will show now, there is a natural induced action of the jet
groupoid $JG$ of~$G$ on the tangent bundle $TE$ of~$E$,
\begin{equation} \label{eq:IACT22}
 \begin{array}{cccc}
  \Phi_{TE}: & JG \times_M TE & \longrightarrow &      TE       \\[1mm]
             &    (u_g,v_e)   &   \longmapsto   & u_g \cdot v_e
 \end{array}
\end{equation}
defined as follows: given $\, g \in G \,$ and $\, e \in E \,$ with
$\, \sigma_G(g) = \pi_E(e)$,  $u_g \in J_g G \,$ and $\, v_e \in T_e E$,
\begin{equation} \label{eq:IACT23}
 u_g \cdot v_e^{}~
 =~T_{(g,e)} \Phi_E \bigl( (u_g \smcirc T_e \pi_E)(v_e) , v_e \bigr) \,.
\end{equation}
This prescription is less obvious than the previous ones because it mixes
the two functors $J$ and $T$, \linebreak so it may be worthwhile to check
explicitly that it does indeed define an action.
To this end, \linebreak we note first that, with $g$, $e$, $u_g$ and $v_e$
as before,
\begin{equation} \label{eq:IACT24}
 \begin{aligned}
  T_{g \cdot e} \pi_E (u_g \cdot v_e)~
  &=~T_{(g,e)} (\pi_E \smcirc \Phi_E)
     \bigl( (u_g \smcirc T_e \pi_E)(v_e) , v_e \bigr) \\
  &=~T_{(g,e)} (\tau_G \smcirc \mathrm{pr}_1)
     \bigl( (u_g \smcirc T_e \pi_E)(v_e) , v_e \bigr) \\
  &=~T_g \tau_G \bigl( (u_g \smcirc T_e \pi_E)(v_e) \bigr)~
   =~(T_g \tau_G \smcirc u_g) (T_e \pi_E(v_e)) \,.
 \end{aligned}
\end{equation}
Therefore, given $\, g,h \in G \,$ and $\, e \in E \,$ with
$\, \sigma_G(g) = \pi_E(e)$, $\sigma_G(h) = \tau_G(g)
= \pi_E(g \cdot e) \,$ and $\, u_g \in J_g G$, $u_h \in J_h G$,
$v_e \in T_e E$,
\[
 \begin{aligned}
  u_h \cdot (u_g \cdot v_e)~
  &=~T_{(h,g \cdot e)} \Phi_E
     \bigl( (u_h \smcirc T_{g \cdot e} \pi_E)(u_g \cdot v_e) \,,
            u_g \cdot v_e \bigr) \\[-0.3ex]
  &=~T_{(h,g \cdot e)} \Phi_E
     \Bigl( (u_h \smcirc (T_g \tau_G \smcirc u_g) \smcirc T_e \pi_E)(v_e) \,,
            T_{(g,e)} \Phi_E \bigl( (u_g \smcirc T_e \pi_E)(v_e) , v_e \bigr)
            \Bigr) \\[-0.3ex]
  &=~T_{(h,g,e)} \bigl( \Phi_E \smcirc (\mathrm{id}_G \times \Phi_E) \bigr)
     \bigl( (u_h \smcirc (T_g \tau_G \smcirc u_g) \smcirc T_e \pi_E)(v_e) \,,
            (u_g \smcirc T_e \pi_E)(v_e) \,, v_e \bigr) \\
  &=~T_{(h,g,e)} \bigl( \Phi_E \smcirc (\mu_G \times \mathrm{id}_E) \bigr)
     \bigl( (u_h \smcirc (T_g \tau_G \smcirc u_g) \smcirc T_e \pi_E)(v_e) \,,
            (u_g \smcirc T_e \pi_E)(v_e) \,, v_e \bigr) \\
  &=~T_{(hg,e)} \Phi_E
     \bigl( \bigl( T_{(h,g)} \mu_G \smcirc
                   \bigl( u_h \smcirc (T_g \tau_G \smcirc u_g) , u_g \bigr)
                   \smcirc T_e \pi_E \bigr)(v_e) \,, v_e \bigr) \\
  &=~T_{(hg,e)} \Phi_E \bigl( (u_hu_g \smcirc T_e \pi_E)(v_e) , v_e \bigr) \\
  &=~(u_h u_g) \cdot v_e \,.
 \end{aligned}
\]
Similarly, given $\, e \in E \,$ with $\, \pi_E(e) = x \,$ and $\, v_e \in T_e E$,
\[
 \begin{aligned}
  1_{JG,x} \cdot v_e~
  &=~T_{(1_{G,x},e)} \Phi_E
     \bigl( (T_x 1_G \smcirc T_e \pi_E)(v_e) \,, v_e \bigr) \\
  &=~T_e \bigl( \Phi_E \smcirc (1_G \smcirc \pi_E \,, \mathrm{id}_E) \bigr)
     (v_e) \\
  &=~T_e \, \mathrm{id}_E (v_e)~=~v_e \,.
 \end{aligned}
\]
Moreover, it follows from equation~(\ref{eq:IACT24}) that this action
preserves the vertical bundle $VE$, and comparing equation~(\ref{eq:IACT02})
or~(\ref{eq:IACT03}) with equation~(\ref{eq:IACT23}) shows that its restriction 
to~$VE$ factorizes through the projection from $JG$ to~$G$ so as to yield the
action~(\ref{eq:IACT01}) introduced above.

A fundamental property of the action~(\ref{eq:IACT22}) is that this one
does respect the structure of $TE$ as a vector bundle over $E$, since
(a) it covers the original action, i.e., the diagram
\begin{equation} \label{eq:IACT25}
 \begin{array}{c}
  \xymatrix{
   JG \times_M TE \ar[rr]^{\qquad \Phi_{TE}^{}} \ar[d] & & TE \ar[d] \\
   G \times_M E \ar[rr]_{\qquad \Phi_E^{}} & & E
  }
 \end{array}
\end{equation}
commutes, and (b)~for any $\, g \in G \,$ and $\, u_g \in J_g G$, left
translation by~$u_g$ is well defined on the whole tangent space $T_e E$
provided that $\, \pi_E(e) = \sigma_G(g) \,$ (and otherwise is not well
defined at any point in $T_e E$), its restriction to this space being
a linear transformation $\, L_{u_g}: T_e E \longrightarrow T_{g \cdot e} E$,
because it is the composition of two linear maps:
\[
 L_{u_g}~=~T_{(g,e)} \Phi_E^{} \,\smcirc\,
           \bigl( u_g \smcirc T_e \pi_E^{} \,, \mathrm{id}_{T_e E} \bigr) \,.
\]
And finally, we note that equation~(\ref{eq:IACT24}) states that the
action~(\ref{eq:IACT22}) also covers the natural action of $GL(TM)$
on~$TM$, i.e., the diagram
\begin{equation} \label{eq:IACT26}
 \begin{array}{c}
  \xymatrix{
   JG \times_M TE \ar[rr]^{\qquad \Phi_{TE}^{}} \ar[d] & & TE \ar[d] \\
   GL(TM) \times_M TM \ar[rr] & & TM
  }
 \end{array}
\end{equation}
commutes.

Another argument showing that the action~(\ref{eq:IACT22}) is the
correct one comes from considering bisections of~$G$ and the auto%
morphisms of~$E$ they generate, according to equation~(\ref{eq:ACTREP1}).
Namely, we can understand this action as the derivative of the push-forward
of curves by automorphisms: given a bisection $\beta$ of~$G$ and a curve
$\gamma$ in~$E$, we set $\, e = \gamma(0)$, $x = \pi_E(e)$ and
$g = \beta(x) \,$ to conclude that if
\begin{equation} \label{eq:IACT27}
 u_g~=~T_x \beta \qquad \mbox{and} \qquad
 v_e~=~\frac{d}{dt} \, \gamma(t) \, \Big|_{t=0} \,,
\end{equation}
then
\begin{equation} \label{eq:IACT28}
 u_g \cdot v_e~
 =~T_e \Pi_E(\beta) (v_e)~
 =~\frac{d}{dt} \, \Pi_E(\beta) (\gamma(t)) \, \Big|_{t=0} \,.
 \vspace*{2mm}
\end{equation}
Actually, repeating the construction at the end of Section~\ref{subsec:LGR},
we can use the action~(\ref{eq:IACT22}) to obtain a representation of the
group $\mathrm{Bis}(JG)$ of bisections of~$JG$ by automorphisms
of~$TE$ (not only as a fiber bundle over~$M$ but also as a vector
bundle over~$E$), that is, a homomorphism
\begin{equation} \label{eq:ACTREP3}
 \begin{array}{cccc}
  \Pi_{TE}^{}:
  & \mathrm{Bis}(JG) & \longrightarrow &      \mathrm{Aut}(TE)      \\[1mm]
  &   \tilde{\beta}  &   \longmapsto   & \Pi_{TE}^{}(\tilde{\beta})
 \end{array}
\end{equation}
which covers the homomorphism $\Pi_E^{}$ defined previously (see equation
(\ref{eq:ACTREP1})) in the following sense: if $\tilde{\beta}$ is holonomous,
say $\, \tilde{\beta} = j\beta$, then equations~(\ref{eq:IACT27}) and~%
(\ref{eq:IACT28}) state that $\Pi_{TE}^{}(\tilde{\beta})$ is the tangent
map to $\Pi_E^{}(\beta)$:
\begin{equation} \label{eq:ACTREP4}
 \Pi_{TE}^{}(j\beta)~=~T \, \Pi_E^{}(\beta) \,.
\end{equation}

The compatibility of the action~(\ref{eq:IACT22}) with the structure of $TE$
as a vector bundle over $E$ allows us to transfer this action to all of its
descendants.
Thus, for example, we obtain an action of $JG$ on tensors of any degree
and type,
\begin{equation} \label{eq:IACT29}
 \begin{array}{cccc}
  \Phi_{T_s^r E}:
  & JG \times_M^{} T_s^r E & \longrightarrow &    T_s^r E    \\[1mm]
  &         (u_g,t_e)      &   \longmapsto   & u_g \cdot t_e
 \end{array}
\end{equation}
and, in particular, on $r$-forms,
\begin{equation} \label{eq:IACT30}
 \begin{array}{cccc}
  \Phi_{\bigwedge^{\!r} T^* E}:
  & JG \times_M^{} \bwedge^{\!r\,} T^* E & \longrightarrow
  & \bwedge^{\!r\,} T^* E \\[1mm]
  &            (u_g,\alpha_e)            &   \longmapsto  
  & u_g \cdot \alpha_e
 \end{array}
\end{equation}
that can be restricted to an action of~$JG$ on partially horizontal $r$-forms,
\begin{equation} \label{eq:IACT31}
 \begin{array}{cccc}
  \Phi_{\bigwedge_s^r T^* E}:
  & JG \times_M^{} \dbwedge{s}{r} T^* E & \longrightarrow
  & \dbwedge{s}{r} T^* E \\[1mm]
  &            (u_g,\alpha_e)           &   \longmapsto  
  & u_g \cdot \alpha_e
 \end{array}
\end{equation}
where $\dbwedge{s}{r} T^* E$ denotes the bundle of $(r-s)$-horizontal
$r$-forms on $E$,%
\addtocounter{footnote}{-6}\footnotemark \addtocounter{footnote}{5}
giving rise to representations of the group of bisections of $JG$, namely
\begin{equation} \label{eq:ACTREP5}
 \begin{array}{cccc}
  \Pi_{T_s^r E}^{}:
  & \mathrm{Bis}(JG) & \longrightarrow &      \mathrm{Aut}(T_s^r E)
  \\[1mm]
  &   \tilde{\beta}  &   \longmapsto   & \Pi_{T_s^ r E}^{}(\tilde{\beta})
 \end{array}
\end{equation}
and, in particular,
\begin{equation} \label{eq:ACTREP6}
 \begin{array}{cccc}
  \Pi_{{\scriptsize \bwedge^{\!r\,}} T^* E}^{}:
  & \mathrm{Bis}(JG) & \longrightarrow
  & \mathrm{Aut}(\bwedge^{\!r\,} T^* E) \\[1mm]
  &   \tilde{\beta}  &   \longmapsto   
  & \Pi_{{\scriptsize \bwedge^{\!r\,}} T^* E}^{}(\tilde{\beta})
 \end{array}
\end{equation}
and
\begin{equation} \label{eq:ACTREP7}
 \begin{array}{cccc}
  \Pi_{{\scriptsize \dbwedge{s}{r}} T^* E}^{}:
  & \mathrm{Bis}(JG) & \longrightarrow 
  & \mathrm{Aut}(\dbwedge{s}{r} T^* E) \\[1mm]
  &   \tilde{\beta}  &   \longmapsto  
  & \Pi_{{\scriptsize \dbwedge{s}{r}} T^* E}^{}(\tilde{\beta})
 \end{array}
\end{equation}

\noindent
As a further consistency check, we note the following:
\begin{prp} \label{prp:ISOINV}~
 Let\/ $E$ be a fiber bundle over a manifold\/~$M$, endowed with the action
 of a Lie groupoid\/ $G$ over the same manifold\/~$M$. Then the canonical 
 strict isomorphism~(\ref{eq:ISOJST1}) is \linebreak $JG$-equivariant.
\end{prp}

The fact that the induced action of $JG$ on $TE$ passes to all its
descendants is the key to understanding what is meant by invariance
of tensor fields under the  action of a Lie groupoid: given a fiber
bundle $E$ over a manifold $M$ endowed with an action of a Lie
groupoid $G$ over the same manifold $M$, invariance of a tensor
field on $E$ refers to the induced action of the jet groupoid $JG$
on the tangent bundle $TE$ and its descendants.
More precisely, a ``pointwise" definition of invariance will involve
some ``stability" Lie subgroupoid $\tilde{G}$ of~$JG$:
\begin{dfn} \label{def:INVTEN1}~
 Let\/ $E$ be a fiber bundle over a manifold\/~$M$, endowed with the action
 of a Lie groupoid\/~$G$ over the same manifold\/~$M$, and let\/ $\tilde{G}$
 be a full Lie subgroupoid of its jet groupoid\/~$JG$. Then we say that a
 tensor field $\, t \in \mathcal{T}_s^r(E) \,$ on\/~$E$ is \textbf%
 {$\tilde{G}$-invariant} if, under the action~(\ref{eq:IACT29}),
 we have
 \begin{equation}
  t_{g \cdot e}~=~u_g \cdot t_e
 \end{equation}
 for all $\, g \in G \,$ and $\, e \in E \,$ such that $\, \sigma_G(g)
= \pi_E(e) \,$ and all $\, u_g \in \tilde{G}_g \subset J_g G$.
\end{dfn}

There is, however, one basic problem with this concept: it is stable
under algebraic operations but, in general, not under operations that
involve differentiation. For example, contraction of a $\tilde{G}$-%
invariant differential form with a $\tilde{G}$-invariant vector field
will produce a $\tilde{G}$-invariant differential form, but the
exterior derivative of a $\tilde{G}$-invariant differential form
will in general no longer be a $\tilde{G}$-invariant differential form.

To see how this comes about and what can be done to remedy the situation,
let us first recast the invariance condition in Definition~\ref{def:INVTEN1}
above in terms of pull-back (or push-forward) under automorphisms generated
by bisections of~$\tilde{G}$. Indeed, it is clear that this invariance condition
is equivalent to condition~(a) of the following
\begin{dfn} \label{def:INVTEN2}~
 Let\/ $E$ be a fiber bundle over a manifold\/~$M$, endowed with the action
 of a Lie groupoid\/~$G$ over the same manifold\/~$M$, and let\/ $\tilde{G}$
 be a full Lie subgroupoid of its jet groupoid\/~$JG$. Given a tensor field $\, t \in
 \mathcal{T}_s^r(E) \,$ on\/~$E$, we say that $t$ is
 \begin{enumerate}[(a)]
  \item \textbf{$\tilde{G}$-invariant} if
        $\, \Pi_{T_s^r E}(\tilde{\beta})_\sharp^{} \, t = t \,$
        for all bisections $\, \tilde{\beta} \in \mathrm{Bis}(\tilde{G})$,
  \item \textbf{$\tilde{G}$-holonomous-invariant} if
        $\, \Pi_{T_s^r E}(\tilde{\beta})_\sharp^{} \, t = t \,$
        for all holonomous bisections $\, \tilde{\beta} \in H\!B(G,\tilde{G})$,
 \end{enumerate}
 where $\,.\,_\sharp$ denotes the push-forward of sections by automorphisms
 as given in equation~(\ref{eq:ACTREP5}).
\end{dfn}
Note that if $\tilde{\beta}$ is holonomous, $\tilde{\beta} = j\beta$, then the
push-forward as defined here coincides with the standard push-forward of
tensor fields on a manifold by diffeomorphisms of the base manifold, i.e.,
$\, \Pi_{T_s^r E}(j\beta)_\sharp \, t = \Pi_E^{}(\beta)_* \, t$.

Now we can make our previous statement about the lack of stability under
differentiation more precise. For the sake of definiteness, let us consider
the case of differential forms and the exterior derivative. Indeed, the
basic property that the exterior derivative $d$ commutes with pull-back
under diffeomorphisms leads immediately to the following
\begin{prp} \label{prp:INVEXD}~
 Let\/ $E$ be a fiber bundle over a manifold\/~$M$, endowed with the action
 of a Lie groupoid\/~$G$ over the same manifold\/~$M$, and let\/ $\tilde{G}$
 be a full Lie subgroupoid of its jet groupoid\/~$JG$. Given a differential
 form $\, \alpha \in \Omega^r(E) \,$ on\/~$E$ which is $\tilde{G}$-holo%
 nomous-invariant, its exterior derivative $\, d\alpha \in \Omega^{r+1}(E) \,$
 is $\tilde{G}$-holonomous-invariant as well.
\end{prp}
It is easy to construct counterexamples showing that the same statement
with ``$\tilde{G}$-holonomous-invariant" replaced by ``$\tilde{G}$-invariant"
is false.

\begin{proof}
 As an exercise to familiarize ourselves with the terminology we are using,
 let us write down explicitly what the different statements of invariance
 mean for differential forms. First, the action (\ref{eq:IACT30}) on
 $r$-forms is explicitly determined from the action (\ref{eq:IACT22})
 on tangent vectors as follows: given $\, g \in G \,$ and $\, e \in E \,$
 with $\, \sigma_G(g) = \pi_E(e)$, $\, u_g \in J_g G$, $\alpha_e \in
 \bwedge^{\!r\,} T_e^* E \,$ and $\, v_1,\ldots,v_r \in T_e^{} E$,
 \begin{equation} \label{eq:IACT32}
  (u_g^{} \cdot \alpha_e^{})
  (u_g^{} \cdot v_1^{} \,, \ldots , u_g^{} \cdot v_r^{})~
  =~\alpha_e^{}(v_1^{} \,, \ldots , v_r^{}) \,.
 \end{equation}
 Thus $\alpha$ will be $\tilde{G}$-invariant iff for any $\, g \in G \,$
 and $\, e \in E \,$ such that $\, \sigma_G(g) = \pi_E(e)$, any
 $\, u_g \in \tilde{G}_g \subset J_g G$ and any tangent vectors
 $\, v_1,\ldots,v_r \in T_e E$, we have
 \begin{equation} \label{eq:INVFOR1}
  \alpha_{g \cdot e} ( u_g \cdot v_1 \,, \ldots , u_g \cdot v_r )~
  =~\alpha_e ( v_1 \,, \ldots , v_r ) \,,
 \end{equation}
 or equivalently, iff, for any bisection $\tilde{\beta}$ of $\tilde{G}$ with
 projected bisection $\,  \beta = \pi_{JG}^{} \smcirc \tilde{\beta} \,$
 of~$G$, any $\, x \in M \,$ and $\, e \in E \,$ such that $\, \pi_E(e)
 = x \,$ and any tangent vectors $\, v_1,\ldots,v_r \in T_e E$, we have
 \begin{equation}
  \alpha_{\beta(x) \cdot e} 
  ( \tilde{\beta}(x) \cdot v_1 \,, \ldots , \tilde{\beta}(x) \cdot v_r )~
  =~\alpha_e ( v_1 \,, \ldots , v_r ) \,,
 \end{equation}
 while it will be $\tilde{G}$-holonomous-invariant iff this is only true
 for any holonomous bisection $\tilde{\beta}$ of~$\tilde{G}$, that is,
 when $\tilde{\beta}$ is reconstructed from $\beta$ as its jet
 prolongation, $\tilde{\beta} = j\beta$, so that the previous
 equation becomes
 \begin{equation}
  \alpha_{\beta(x) \cdot e} 
  ( T_x \beta \cdot v_1 \,, \ldots , T_x \beta \cdot v_r )~
  =~\alpha_e ( v_1 \,, \ldots , v_r ) \,.
 \end{equation}
 This makes it clear that the condition of being $\tilde{G}$-holonomous-%
 invariant is stable under exterior differentiation, but that of being
 $\tilde{G}$-invariant is not.
 \qed
\end{proof}

The final problem that remains is to convert the property of being $\tilde{G}$-%
holonomous-invariant back into a ``pointwise defined" condition of invariance
under some Lie groupoid over~$M$. This may be possible in some cases and
impossible in others, since the group of holonomous bisections of $\tilde{G}$
may or may not be equal to the group of all bisections of some Lie groupoid
over~$M$: after all, the constraint that a bisection should be holonomous
involves not the values of that bisection but rather its first order derivatives,
through the Frobenius integrability condition. So this question will have to be
handled case by case.

With all these preliminaries out of the way, we are finally ready to
show in which sense the multisymplectic structure of classical field
theory is invariant:
\begin{thm}~ \label{te:MSPINV}
 Let\/ $E$ be a fiber bundle over a manifold\/~$M$, endowed with the
 action of a Lie groupoid\/~$G$ over the same manifold\/~$M$, and
 consider the induced actions of\/~$JG$ on the extended multiphase
 space\/ $J^{\circledstar} E$ and of the second order jet groupoids
 $\, J^2 G \subset \bar{J}^{\,2} G \subset J(JG) \,$ on its tangent
 bundle\/ $T(J^{\circledstar} E)$ and its descendants.
 Then the multicanonical form\/~$\theta$ is invariant under the action
 of the semiholonomous second order jet groupoid\/~$\bar{J}^{\,2} G$,
 whereas the multisymplectic form\/~$\omega$ is only invariant under
 the action of the second order jet groupoid\/~$J^2 G$.
%
\end{thm}
\begin{proof}
 Initially, we remember that the induced actions of~$JG$ on~$JE$ and on~$TE$
 cover the original action of~$G$ on~$E$ (see the diagrams~(\ref{eq:IACT10})
 and~(\ref{eq:IACT25})), and hence the same holds for the induced action of
 $JG$ on the extended multiphase space of equation~(\ref{eq:ISOJST1}), which
 for the sake of brevity we shall here denote by $\Lambda$.
 In other words, the diagram
 \begin{equation} \label{eq:IACT33}
  \begin{array}{c}
   \xymatrix{
    JG \times_M \Lambda \ar[rr]^{\qquad \Phi_\Lambda^{}}
    \ar[d]_{(\pi_{JG}^{},\pi_\Lambda^{})} & &
    \Lambda \ar[d]^{\pi_\Lambda^{}} \\
    G \times_M E \ar[rr]_{\qquad \Phi_E^{}} & & E
   }
  \end{array}
 \end{equation}
 commutes, that is,
 \vspace{1ex}
 \[
  \pi_\Lambda \smcirc\, \Phi_\Lambda
  = \Phi_E \smcirc (\pi_{JG},\pi_\Lambda) \,.
 \vspace{1ex}
 \]
 Therefore, given $\, u_g \in J_g G \,$ and $\, \alpha_e \in \Lambda_e \,$ with
 $\, \sigma_G(g) = \pi_E(e)$, $u'_{u_g} \in \bar{J}_{u_g}^{\,2} G \,$ and 
 $\, w \in T_{\alpha_e} \Lambda$, we have by equation~(\ref{eq:IACT23}),
 \[
  \begin{aligned}
   T_{u_g \cdot \alpha_e}^{} \pi_\Lambda^{}(u'_{u_g} \!\cdot w)~
   &=~(T_{u_g \cdot \alpha_e}^{} \pi_\Lambda^{} \smcirc
       T_{(u_g,\alpha_e)}^{} \Phi_\Lambda^{})
      \bigl( (u'_{u_g} \!\smcirc T_{\alpha_e}^{}
              (\pi_E^{} \smcirc \pi_{\Lambda}^{}))(w) \,, w \bigr) \\[0.5ex]
   &=~\bigl( T_{(g,e)} \Phi_E^{} \smcirc
             (T_{u_g} \pi_{JG}^{},T_{\alpha_e}^{} \pi_\Lambda^{}) \bigr)
      \bigl( (u'_{u_g} \!\smcirc T_{\alpha_e}^{}
              (\pi_E^{} \smcirc \pi_{\Lambda}^{}))(w) \,, w \bigr) \\
   &=~T_{(g,e)} \Phi_E^{}
      \Bigl( \bigl( T_{u_g} \pi_{JG}^{} \smcirc u'_{u_g} \!\smcirc
                    T_{\alpha_e} (\pi_E^{} \smcirc \pi_{\Lambda}^{}) \bigr)
             (w) \,, T_{\alpha_e} \pi_{\Lambda}^{}(w) \Bigr) \\
   &=~T_{(g,e)} \Phi_E^{}
      \bigl( (u_g^{} \smcirc T_e^{} \pi_E^{})
             (T_{\alpha_e} \pi_{\Lambda}^{}(w)) \,,
             T_{\alpha_e} \pi_{\Lambda}^{}(w) \bigr) \\[0.5ex]
   &=~u_g \cdot (T_{\alpha_e} \pi_{\Lambda}^{}(w)) \,.
  \end{aligned}
 \]
 Thus, given $\, w_1,\ldots,w_n \in T_{\alpha_e} \Lambda$, we get
 \[
  \begin{aligned}
   &\theta_{u_g \cdot \alpha_e}^{}
    \bigl( u'_{u_g} \!\cdot w_1^{} \,,\ldots, u'_{u_g} \!\cdot w_n^{} \bigr) \\
   &~=~(u_g^{} \cdot \alpha_e^{})
       \bigl( T_{u_g \cdot \alpha_e}^{} \pi_\Lambda^{}
              (u'_{u_g} \!\cdot w_1^{}) \,, \ldots,
              T_{u_g \cdot \alpha_e}^{} \pi_\Lambda^{}
              (u'_{u_g} \!\cdot w_n^{}) \bigr) \\
   &~=~(u_g \cdot \alpha_e)
       \bigl( u_g \cdot (T_{\alpha_e} \pi_\Lambda^{} (w_1)) \,, \ldots,
              u_g \cdot (T_{\alpha_e} \pi_\Lambda^{} (w_n)) \bigr) \\
   &~=~\alpha_e
       \bigl( T_{\alpha_e} \pi_\Lambda^{} (w_1) \,, \ldots,
              T_{\alpha_e} \pi_\Lambda^{} (w_n) \bigr) \\
   &~=~\theta_{\alpha_e} \bigl( w_1^{} \,,\ldots, w_n^{} \bigr) \,,
  \end{aligned}
 \]
 proving that $\theta$ is $\bar{J}^{\,2} G$\,-\,invariant.
 To prove that that $\, \omega = - d\theta \,$ is $J^2 G$\,-\,invariant,
 we shall apply Proposition~\ref{prp:INVEXD}, but in order to do so,
 we have to convert both of these "pointwise" invariance statements to
 statements of invariance under pull-back with automorphisms generated
 by the appropriate group of holonomous bisections.
 Now as we have seen above, $\bar{J}^{\,2} G$\,-\,invariance of~%
 $\theta$ is equivalent to invariance of~$\theta$ under pull-back with
 automorphisms generated by arbitrary bisections of~$\bar{J}^{\,2} G$,
 which of course trivially implies that the same property holds for
 holonomous bisections of~$\bar{J}^{\,2} G$, and this property
 is what extends from $\theta$ to~$\omega$, according to
 Proposition~\ref{prp:INVEXD}.
 But as we have seen above (see equation~(\ref{eq:HOLBIS})),
 the holonomous bisections of~$\bar{J}^{\,2} G$ are precisely the
 bisections of~$J^2 G$, which concludes the proof.
 \qed
\end{proof}

It may seem strange that $\theta$ has a higher degree of symmetry
than~$\omega$ since in various other contexts, one faces a reverse
situation where one encounters a structure form $\omega$ which is
invariant whereas its "potential" $\theta$ is only invariant up to
addition of an exact form.
But that's the way things turn out here.

The next step is to formulate the concept of invariance of the lagrangian
and/or hamiltonian under a Lie groupoid action:
\begin{dfn} \label{def:INVLH}~
 Let\/ $E$ be a fiber bundle over a manifold\/~$M$, endowed with the action of
 a Lie groupoid\/~$G$ over the same manifold\/~$M$, and consider the induced
 actions of\/~$JG$ on the jet bundle\/ $JE$, the extended multiphase space\/
 $J^{\circledstar} E$ and the ordinary multiphase space\/
 $\vec{J}^{\,\circledast} E$.
 Given a full Lie subgroupoid\/~$\tilde{G}$ of\/~$JG$, we say that a
 lagrangian $\, \mathcal{L}: JE \longrightarrow \bwedge^{\!n\,} T^* M \,$
 and/or a hamiltonian $\, \mathcal{H}: \vec{J}^{\,\circledast} E
 \longrightarrow J^{\circledstar} E \,$ is\/ \textbf{$\tilde{G}$-invariant}
 if it is equivariant with respect to the pertinent actions of\/~$JG$, when
 restricted to\/~$\tilde{G}$, i.e., if the diagrams
 \begin{equation} \label{eq:INVLAG}
  \begin{array}{c}
   \xymatrix{
    \tilde{G} \times_M JE \ar[rr]^{\qquad \Phi_{JE}^{}}
    \ar[d]_{(\pi_{JG}^{\mathrm{fr}}|_{\tilde{G}},\mathcal{L})}
    & & JE  \ar[d]^{\mathcal{L}} \\
    GL(TM) \times_M \bwedge^{\!n\,} T^* M \ar[rr]
    & & \bwedge^{\!n\,} T^* M
   }
  \end{array}
 \end{equation} 
 and/or
 \begin{equation} \label{eq:INVHAM}
  \begin{array}{c}
   \xymatrix{
    \tilde{G} \times_M \vec{J}^{\,\circledast} E
    \ar[rr]^{\qquad \Phi_{\vec{J}^{\,\circledast} E}^{}}
    \ar[d]_{(\mathrm{id}_{\tilde{G}}^{},\mathcal{H})} & &
    \vec{J}^{\,\circledast} E
    \ar[d]^{\mathcal{H}_{\vphantom{\tilde{G}}}} \\
    \tilde{G} \times_M J^{\circledstar} E
    \ar[rr]_{\qquad \Phi_{J^{\circledstar} E}} & &
    J^{\circledstar} E \vphantom{\tilde{G}}
   }
  \end{array}
 \end{equation}
 commute.
\end{dfn}
It is noteworthy that $\tilde{G}$-invariance of a lagrangian
$\, \mathcal{L}: JE \longrightarrow \bwedge^{\!n\,} T^* M \,$
implies $\tilde{G}$-invariance of its Legendre transformation
$\, \mathbb{F} \mathcal{L}: JE \longrightarrow J^{\circledstar} E$,
in the sense that the following diagram commutes:
\begin{equation} \label{eq:INVLEG}
 \begin{array}{c}
  \xymatrix{
   \tilde{G} \times_M JE
   \ar[rr]^{\qquad \Phi_{JE}^{}}
   \ar[d]_{(\mathrm{id}_{\tilde{G}}^{},\mathbb{F} \mathcal{L})} & &
   JE
   \ar[d]^{\mathbb{F} \mathcal{L}_{\vphantom{\tilde{G}}}} \\
   \tilde{G} \times_M J^{\circledstar} E
   \ar[rr]_{\qquad \Phi_{J^{\circledstar} E}} & &
   J^{\circledstar} E \vphantom{\tilde{G}}
  }
 \end{array}
\end{equation}
Indeed, given $\, g \in G \,$ with $\, \sigma_G(g) = x \,$ and
$\, \tau_G(g) = y$, $e \in E \,$ with $\, \pi_E(e) = x$, $u_g \in
\tilde{G}_g \subset J_g G \,$ and $\, u_e^{},u'_e \in J_e E$,
we get, according to equations~(\ref{eq:LEGT4}) and~%
(\ref{eq:IACT13}),
\[
 \begin{aligned}
   &\mathbb{F} \mathcal{L}(u_g \cdot u_e^{}) \cdot (u_g \cdot u'_e) \\
   &~=~\mathcal{L}(u_g \cdot u_e^{}) \, + \,
       \frac{d}{dt} \, \mathcal{L}
       \bigl( u_g \cdot u_e^{} + t (u_g \cdot u_e' - u_g \cdot u_e^{}) \bigr)
       \Big|_{t=0} \\
   &~=~\mathcal{L}(u_g \cdot u_e^{}) \, + \,
       \frac{d}{dt} \, \mathcal{L}
       \bigl( u_g \cdot (u_e^{} + t (u_e' - u_e^{})) \bigr) \Big|_{t=0} \\
   &~=~(T_g \tau_G^{} \smcirc u_g)^{-1^{\,\scriptstyle{*}}}
       \mathcal{L}(u_e^{}) \, + \, \frac{d}{dt} \, 
       (T_g \tau_G^{} \smcirc u_g)^{-1^{\,\scriptstyle{*}}}
       \mathcal{L} \bigl( u_e^{} + t (u_e' - u_e^{}) \bigr) \Big|_{t=0} \\
   &~=~(T_g \tau_G^{} \smcirc u_g)^{-1^{\,\scriptstyle{*}}}
       \Bigl( \mathcal{L}(u_e^{}) \, + \, \frac{d}{dt} \,
              \mathcal{L} \bigl( u_e^{} + t (u_e' - u_e^{}) \bigr) \Big|_{t=0}
              \Bigr) \\[1ex]
   &~=~(T_g \tau_G^{} \smcirc u_g)^{-1^{\,\scriptstyle{*}}}
       \bigl( \mathbb{F} \mathcal{L}(u_e^{}) \cdot u'_e \bigr) \\[2ex]
   &~=~(u_g \cdot \mathbb{F} \mathcal{L}(u_e^{})) \cdot
       (u_g \cdot u'_e) \,.
 \end{aligned}
\]
Therefore, we expect invariance of the lagrangian or hamiltonian to ensure
a corresponding form of invariance of the forms $\theta_{\mathcal{L}}$,
$\omega_{\mathcal{L}}$ and $\theta_{\mathcal{H}}$,
$\omega_{\mathcal{H}}$ defined by pull-back,
\[
 \theta_{\mathcal{L}}^{}~
 =~(\vec{\mathbb{F}} \mathcal{L})^* \theta_{\mathcal{H}}^{}~
 =~({\mathbb{F}} \mathcal{L})^* \theta~~,~~
 \omega_{\mathcal{L}}^{}~
 =~(\vec{\mathbb{F}} \mathcal{L})^* \omega_{\mathcal{H}}^{}~
 =~({\mathbb{F}} \mathcal{L})^* \omega \,,
\]
(see, e.g., equation~(\ref{eq:MULTCS}) and also equation~%
(\ref{eq:POCARF})), but the resulting invariance is reduced:
\begin{thm} \label{te:LHINV}~
 Let\/ $E$ be a fiber bundle over a manifold\/~$M$, endowed with the
 action of a Lie groupoid\/~$G$ over the same manifold\/~$M$, and
 consider the induced actions of\/~$JG$ on the jet bundle\/~$JE$,
 the extended multiphase space\/ $J^{\circledstar} E$ and the
 ordinary multiphase space\/ $\vec{J}^{\,\circledast} E$, as well as
 of the second order jet groupoids $\, J^2 G \subset \bar{J}^{\,2} G
 \subset J(JG) \,$ on the respective tangent bundles and their descendants.
 Given a full Lie subgroupoid\/~$\tilde{G}$ of\/~$JG$, suppose that the
 lagrangian $\, \mathcal{L}: JE \longrightarrow \bwedge^{\!n\,} T^* M \,$
 and/or hamiltonian $\, \mathcal{H}: \vec{J}^{\,\circledast} E \longrightarrow
 J^{\circledstar} E \,$ are\/~$\tilde{G}$-invariant. 
 Then the forms\/ $\theta_{\mathcal{L}}$, $\omega_{\mathcal{L}}$
 and/or $\theta_{\mathcal{H}}$, $\omega_{\mathcal{H}}$ are
 $\tilde{G}$-holonomous-invariant.
\end{thm}
\begin{proof}
 We begin by reformulating the commutativity of the diagram in
 equation~(\ref{eq:INVHAM}) as stating that, for any bisection
 $\tilde{\beta}$ of $\tilde{G}$,
 \[
  \mathcal{H} \smcirc \Pi_{\vec{J}^{\,\circledast} E}(\tilde{\beta})~
  =~\Pi_{J^{\circledstar} E}(\tilde{\beta}) \smcirc \mathcal{H} \,.
 \]
 Thus we conclude from Theorem~\ref{te:MSPINV} that
 \[
  \Pi_{\vec{J}^{\,\circledast} E}(\tilde{\beta})^*
  \bigl( \mathcal{H}^* \theta \bigr)~
  =~\mathcal{H}^* \bigl( \Pi_{J^{\circledstar} E}
    (\tilde{\beta})^* \theta \bigr)~=~\mathcal{H}^* \theta
 \]
 for any bisection $\tilde{\beta}$ of~$\tilde{G}$ such that
 $j \tilde{\beta}$ takes values in $\bar{J}^{\,2} G$, and
 \[
  \Pi_{\vec{J}^{\,\circledast} E}(\tilde{\beta})^*
  \bigl( \mathcal{H}^* \omega \bigr)~
  =~\mathcal{H}^* \bigl( \Pi_{J^{\circledstar} E}
    (\tilde{\beta})^* \omega \bigr)~=~\mathcal{H}^* \omega
 \]
 for any holonomous bisection $\tilde{\beta}$ of~$\tilde{G}$
 such that $j \tilde{\beta}$ takes values in $\bar{J}^{\,2} G$.
 But as we have seen before, a bisection $\tilde{\beta}$ of~%
 $\tilde{G}$ (or even of~$JG$) such that $j \tilde{\beta}$ takes
 values in $\bar{J}^{\,2} G$ is necessarily holonomous,
 $\tilde{\beta} = j \beta$ (where $\, \beta = \pi_{JG} \circ
 \tilde{\beta}$), so $\, \tilde{\beta} \in H\!B(G,\tilde{G})$,
 The proof in the lagrangian context is the same, just replacing
 equation~(\ref{eq:INVHAM}) by equation~(\ref{eq:INVLAG}),
 $\vec{J}^{\,\circledast} E$ by $JE$ and $\mathcal{H}$ by
 $\mathbb{F} \mathcal{L}$.
 \hspace*{\fill} \qed
\end{proof}

\section{Noether's Theorem}

We begin this section introducing the concept of momentum map in the
context of Lie groupoid actions: it comes in two variants, depending
on whether we work in extended or in ordinary multi\-phase space.
To prepare the ground, let us formulate the infinitesimal versions
of the Definitions~\ref{def:FULLSG} and~\ref{def:HOLBIS}:
\begin{dfn} \label{def:FULLSA}~
 Let\/ $\mathfrak{g}$ be a Lie algebroid over a manifold\/~$M$.
 We say that a Lie subalgebroid\/~$\tilde{\mathfrak{g}}$ of its jet
 algebroid\/~$J\mathfrak{g}$ is \textbf{full} if the projection
 of\/~$J\mathfrak{g}$ to\/~$\mathfrak{g}$ remains a bundle
 projection when restricted to\/~$\tilde{\mathfrak{g}}$.
\end{dfn}
\begin{dfn} \label{def:HOLSEC}~
 Let\/ $\mathfrak{g}$ be a Lie algebroid over a manifold\/~$M$ and\/
 $\tilde{\mathfrak{g}}$ a full Lie subalgebroid of its jet algebroid\/
 $J\mathfrak{g}$.
 Then we define the \textbf{Lie algebra of holonomous sections}
 of\/~$\tilde{\mathfrak{g}}$, denoted by\/ $H\!\varGamma(\mathfrak{g},
 \tilde{\mathfrak{g}})$, to be the Lie subalgebra of\/~$\varGamma%
 (\tilde{\mathfrak{g}})$ given by
 \begin{equation}
  H\!\varGamma(\mathfrak{g},\tilde{\mathfrak{g}})~
  =~\{ \tilde{X} \in \varGamma(\tilde{\mathfrak{g}})~|~\tilde{X} = jX~
       \mbox{for some $\, X \in \varGamma(\mathfrak{g})$} \} \,,
 \end{equation}
 which can also be viewed as a Lie subalgebra of\/~$\varGamma(\mathfrak{g})$,
namely, the one given by
 \begin{equation}
  H\!\varGamma(\mathfrak{g},\tilde{\mathfrak{g}})~
  \cong~\bigl\{ X \in \varGamma(\mathfrak{g})~|~
        jX \in \varGamma(\tilde{\mathfrak{g}}) \bigr\} \,.
 \end{equation}
\end{dfn}
In what follows, we shall often switch between these two
interpretations without further mention.
\begin{dfn} \label{def:MOMAP}~
 Let\/ $E$ be a fiber bundle over a manifold\/ $M$, endowed with the action of
 a Lie groupoid\/ $G$ over the same manifold\/ $M$, and consider the induced
 actions of\/~$JG$ on the extended multiphase space\/ $J^{\circledstar} E$
 and the ordinary multiphase space\/ $\vec{J}^{\,\circledast} E$, as well as
 the corresponding infinitesimal actions of the Lie algebroids\/~$\mathfrak{g}$
 (by fundamental vector fields on\/~$E$) and\/ $J\mathfrak{g}$ (by fundamental
 vector fields on\/~$J^{\circledstar} E$ and\/~$\vec{J}^{\,\circledast} E$).
 Then the \textbf{extended momentum map}\/ $\mathcal{J}^{\mathrm{ext}}$
 associated to each of these actions is the map
 \begin{equation} \label{eq:MMAPE1}
  \mathcal{J}^{\mathrm{ext}}:~\varGamma(J\mathfrak{g})~
  \longrightarrow~\Omega^{n-1} \bigl( J^{\circledstar} E \bigr)
 \end{equation}
 defined by
 \begin{equation} \label{eq:MMAPE2}
  \mathcal{J}^{\mathrm{ext}}(Z) \, 
  = \, i_{Z_{J^{\circledstar} E}}^{} \theta \,,
 \end{equation}
 and the map
 \begin{equation} \label{eq:MMAPO1}
  \mathcal{J}^{\mathrm{ext}}:~\varGamma(J\mathfrak{g})~
  \longrightarrow~\Omega^{n-1} \bigl( \vec{J}^{\,\circledast} E \bigr)
 \end{equation}
 defined by
 \begin{equation} \label{eq:MMAPO2}
  \mathcal{J}^{\mathrm{ext}}(Z) \,
  = \, i_{Z_{\vec{J}^{\,\circledast} E}}^{} \theta_{\mathcal{H}} \,,
 \end{equation}
 respectively, and the corresponding \textbf{momentum map} is its
 composition with the jet prolongation map from\/ $\varGamma(\mathfrak{g})$
 to\/~$\varGamma(J\mathfrak{g})$, so
 \begin{equation} \label{eq:MMAPE3}
  \mathcal{J}:~\varGamma(\mathfrak{g})~
  \longrightarrow~\Omega^{n-1} \bigl( J^{\circledstar} E \bigr)
 \end{equation}
 with
 \begin{equation} \label{eq:MMAPE4}
  \mathcal{J}(X) \, 
  = \, i_{X_{J^{\circledstar} E}}^{} \theta \,,
 \end{equation}
 and
 \begin{equation} \label{eq:MMAPO3}
  \mathcal{J}:~\varGamma(\mathfrak{g})~
  \longrightarrow~\Omega^{n-1} \bigl( \vec{J}^{\,\circledast} E \bigr)
 \end{equation}
 with
 \begin{equation} \label{eq:MMAPO4}
  \mathcal{J}(X) \,
  = \, i_{X_{\vec{J}^{\,\circledast} E}}^{} \theta_{\mathcal{H}} \,,
 \end{equation}
 where\/ $X_{J^{\circledstar} E}$ and\/ $X_{\vec{J}^{\,\circledast} E}$
 are the canonical (dualized jet) lifts of\/ $X_E$ from\/~$E$ to\/
 $J^{\circledstar} E$ and\/ $\vec{J}^{\,\circledast} E$, which coincide
 with the fundamental vector fields\/ $(jX)_{J^{\circledstar} E}$ and\/
 $(jX)_{\vec{J}^{\,\circledast} E}$, respectively. 
\end{dfn}
Only the ordinary multiphase space version appears directly in Noether's
theorem:
\begin{thm}[Noether's theorem] \label{te:NOETH}~
 Let\/ $E$ be a fiber bundle over a manifold\/~$M$, endowed with the action
 of a Lie groupoid\/~$G$ over the same manifold\/~$M$, and consider the
 induced action of\/~$JG$ on the ordinary multiphase space\/ $\vec{J}%
 ^{\,\circledast} E$, as well as the corresponding infinitesimal actions of
 the Lie algebroids\/ $\mathfrak{g}$ (by fundamental vector fields on\/~$E$)
 and\/ $J\mathfrak{g}$ (by fundamental vector fields on\/~$\vec{J}%
 ^{\,\circledast} E$).
 Given a full Lie subgroupoid\/ $\tilde{G}$ of\/~$JG$, with corresponding
 full Lie subalgebroid\/ $\tilde{\mathfrak{g}}$ of\/~$J\mathfrak{g}$, and a\/
 $\tilde{G}$-invariant hamiltonian $\, \mathcal{H}: \vec{J}^{\,\circledast} E
 \longrightarrow J^{\circledstar} E$, the \textbf{Noether current} associated
 with a ``generator'' $\, X \in H\!\varGamma(\mathfrak{g},\tilde{\mathfrak{g}}) \,$
 and a section\/~$\phi$ of\/~$\vec{J}^{\,\circledast} E \,$ is the pull-back
 $\, \phi^* \mathcal{J}(X) \in \Omega^{n-1}(M)$.
 Then if $\phi$ satisfies the equations of motion, i.e., the De\,Donder\,--\,%
 Weyl equations, this current is \textbf{conserved}, i.e., a closed form:
 \[
  d[\phi^* \mathcal{J}(X)] \, = \, 0 \,.
 \]
\end{thm}
\begin{proof}
 Given $\, X \in H\!\varGamma(\mathfrak{g},\tilde{\mathfrak{g}})$ and
 $\, \phi \in \varGamma(\vec{J}^{\,\circledast} E)$, we have
 \[
  \begin{aligned}
   d[\phi^* \mathcal{J}(X)]~
   &=~d[\phi^* (i_{X_{\vec{J}^{\,\circledast} E}} \theta_{\mathcal{H}})]~
    =~\phi^* d (i_{X_{\vec{J}^{\,\circledast} E}} \theta_{\mathcal{H}}) \\
   &=~\phi^* (L_{X_{\vec{J}^{\,\circledast} E}} \theta_{\mathcal{H}}) \,
      + \, \phi^* (i_{X_{\vec{J}^{\,\circledast} E}} \omega_\mathcal{H}) \,.
  \end{aligned}
 \]
 \emph{Claim 1:}
 \[
  L_{X_{\vec{J}^{\,\circledast} E}}\theta_\mathcal{H} = 0 \,.
 \]
 By hypothesis, $X$ generates a one-parameter subgroup of holonomous
 bisections $\exp(tX)$ of~$\tilde{G}$ and $X_{\vec{J}^{\,\circledast} E}$
 generates the one-parameter subgroup $\Pi_{\vec{J}^{\,\circledast} E}
 (j(\exp(tX)))$ of automorphisms of~$\vec{J}^{\,\circledast} E$, which
 just as in the proof of Theorem~\ref{te:LHINV} implies that
 \[
  \Pi_{\vec{J}^{\,\circledast} E}(j(\exp(tX)))^* \,
  \theta_{\mathcal{H}}~=~\theta_{\mathcal{H}}
 \]
 and hence
 \[
  L_{X_{\vec{J}^{\,\circledast} E}} \theta_{\mathcal{H}}~
  =~\frac{d}{dt} \; \Pi_{\vec{J}^{\,\circledast} E}(j(\exp(tX)))^* \,
    \theta_{\mathcal{H}} \, \Big|_{t=0} \,
  =~\frac{d}{dt} \, \theta_{\mathcal{H}} \, \Big|_{t=0}~=~0 \,.
 \]
 \emph{Claim 2:} If $\phi$ is a solution of the De\,Donder\,--\,Weyl equations,
 then since $X_{\vec{J}^{\,\circledast} E}$ is projectable, it follows from
 equation~(\ref{eq:HEQMOT}) that
 \[
  \phi^* (i_{X_{\vec{J}^{\,\circledast} E}} \omega_{\mathcal{H}})~=~0 \,.
 \]
 \vspace{-2ex}
 \qed
\end{proof}

\section{Example: Theory of a real scalar field}

In this section, we want to illustrate the concepts and constructions intro%
duced in this paper on what is perhaps the simplest possible example: the
theory of a single real scalar field on a Lorentz manifold $M$ with metric
tensor $\mathslf{g}$.
Here, the situation is substantially simplified because this theory has no
(continuous) internal symmetries; yet groupoids are still relevant to handle
its space-time symmetries.
The configuration bundle $E$ is the trivial real line bundle over~$M$,
whose bundle projection $\pi_E^{}$ is the projection $\mathrm{pr}_1^{}$
onto the first factor and whose (smooth) sections are just ordinary
(smooth) functions on~$M$,
\begin{equation}
 E~=~M \times \mathbb{R}~~,~~
 \varGamma(E)~=~C^\infty(M,\mathbb{R}) \,.
\end{equation}
Also, the basic Lie groupoid $G$ for symmetry considerations is the pair
groupoid of~$M$, whose source projection $\sigma_G^{}$ and target
projection $\tau_G^{}$ are the projections onto the second and the
first factor, respectively, and whose (smooth) bisections are just the
(smooth) diffeomorphisms of~$M$,
\begin{equation}
 G~=~M \times M~~,~~
 \mathrm{Bis}(G)~=~\mathrm{Diff}(M) \,,
\end{equation}
and hence, as is well known, the corresponding Lie algebroid $\mathfrak{g}$
is the tangent bundle of~$M$, whose (smooth) sections are just the
(smooth) vector fields on~$M$,
\begin{equation}
 \mathfrak{g}~=~TM~~,~~
 \varGamma(\mathfrak{g})~=~\mathfrak{X}(M) \,,
\end{equation}
so that the exponential map in equation~(\ref{eq:EXP1}) is the
standard one that associates to each vector field its flow.%
\footnote{For simplicity of presentation, we disregard questions of
completeness here.}
Next, passing to first order jets, we note that since $E$ is globally trivial,
we can identify its jet bundle with its linearized jet bundle: both are
just the pull-back of the cotangent bundle of~$M$ to~$E$,
\begin{equation}
 JE~=~\vec{J} E~=~\mathrm{pr}_1^* \bigl( T^\ast M \bigr) \,.
\end{equation}
This leads to the following identifications for the ordinary and extended
multiphase spaces:
\begin{equation}
 \vec{J}^{\,\circledast} E
 =~\mathrm{pr}_1^* \bigl( \bwedge^{\!n-1\,} T^\ast M \bigr)
 \quad , \quad
 J^{\circledstar} E~
 =~\mathrm{pr}_1^* \bigl( \bwedge^{\!n-1\,} T^\ast M \oplus
                          \bwedge^{\!n\,} T^\ast M \bigr) \,.
\end{equation}
Similarly, as already stated in Example~\ref{ex:LFRGR}, the jet groupoid
of the pair groupoid is the linear frame groupoid,
\begin{equation}
 JG~=~GL(TM) \,,
\end{equation}
and hence the jet algebroid of the tangent bundle is the linear frame
algebroid,
\begin{equation}
 J\mathfrak{g}~=~\mathfrak{gl}(TM) \,,
\end{equation}
which as a vector bundle is the direct sum $\, TM \oplus L(TM) \,$ (see
equation~(\ref{eq:JETALG2})) but whose bracket operation we shall not
specify explicitly since we shall not need it here.

With these preliminaries out of the way, we can specify the lagrangian and
the (covariant) hamiltonian of the theory: $\mathcal{L} = L \, d^{\,n} x$
and $\mathcal{H} = - H \, d^{\,n} x$ where the functions $L$ on $JE$ and
$H$ on $\vec{J}^{\,\circledast} E$ are given by
\begin{equation} \label{eq:LAGSF}
 L(\varphi,d\varphi)~=~\tfrac{1}{2} \, \mathslf{g}^{\mu\nu} \,
 \partial_\mu^{} \varphi \, \partial_\nu^{} \varphi \, - \, V(\varphi) \,,
\end{equation}
and
\begin{equation} \label{eq:HAMSF}
 H(\varphi,\pi)~=~\tfrac{1}{2} \, \mathslf{g}_{\mu\nu}^{} \,
 \pi^\mu \pi^\nu \, + \, V(\varphi) \,,
\end{equation}
respectively, where $\, d\varphi = \partial_\mu^{} \varphi \, dx^\mu \,$
and $\, \pi = \pi^\mu \, d^{\,n} x_\mu^{} \,$, and $V$ is some potential.
Obviously, then, the symmetry groupoid for this theory is the orthonormal
frame groupoid with respect to the metric~$\mathslf{g}$,
\begin{equation}
 \tilde{G}~=~O(TM,\mathslf{g}) \,.
\end{equation}
Thus the group of holonomic bisections of~$\tilde{G}$ is precisely the
isometry group of~$(M,\mathslf{g})$,
\begin{equation}
 H\!B(G,\tilde{G})~=~\mathrm{Isom}(M,\mathslf{g}) \,,
\end{equation}
and the Lie algebra of holonomous sections of $\tilde{\mathfrak{g}}$ is
precisely the Lie algebra of Killing vector fields on~$(M,\mathslf{g})$,
\begin{equation}
 H\!\varGamma(\mathfrak{g},\tilde{\mathfrak{g}})~
 =~\mathrm{Kill}(M,\mathslf{g}) \,.
\end{equation}
Given such a Killing vector field~$K$ on~$(M,\mathslf{g})$, and writing
\[
 K~=~K^\mu \, \frac{\partial}{\partial x^\mu} \,,
\]
we first see that its lift $K_E$ to~$E$ has the same form (the component
in the fiber direction of~$E$, i.e., along the vertical vector field
$\partial/\partial q$, is zero, which reflects the fact that we are
dealing with a scalar field theory, where the field values are
invariant under diffeomorphisms of~$M$) and hence its canonical
lift to $\vec{J}^{\,\circledast} E$ reads
\[
 K_{\vec{J}^{\,\circledast} E}~
 =~K^\mu \, \frac{\partial}{\partial x^\mu} \, + \,
   \Bigl( \frac{\partial K^\mu}{\partial x^\nu} \, p^\nu \,- \,
          \frac{\partial K^\nu}{\partial x^\nu} \, p^\mu \Bigr) \,
   \frac{\partial}{\partial p^\mu} \,,
\]
which, when contracted with
\[
 \theta_{\mathcal{H}}~
 =~p^\mu \; dq \wedge d^{\,n} x_\mu^{} \, - \, H \, d^{\,n} x~,
\]
gives
\[
 \mathcal{J}(K)~=~i_{K_{\vec{J}^{\,\circledast} E}} \theta_{\mathcal{H}}~
 = \; - \,  p^\mu K^\nu \; dq \wedge d^{\,n} x_{\mu\nu}^{} \, - \,
      H \, K^\mu \, d^{\,n} x_\mu^{} \,.
\]
Pulling back with a field configuration $\, \phi = (\varphi,\pi)$ amounts to
substituting $q$ by $\varphi$, $dq$ by $d\varphi = \partial_\kappa^{}
\varphi \; dx^\kappa \,$ and $p^\mu$ by $\, \pi^\mu = g^{\mu\nu} \,
\partial_\nu^{} \varphi$, and using $\, dx^\kappa \wedge
d^{\,n} x_{\mu\nu}^{} = \delta_\nu^\kappa \, d^{\,n} x_\mu^{}
\, - \, \delta_\mu^\kappa \, d^{\,n} x_\nu^{}$, we get
\begin{equation}
 \phi^* \mathcal{J}(X)~
 = \; - \, T_{\mu\nu}^{} \, K^\mu \mathslf{g}^{\nu\kappa} \,
   d^{\,n} x_\kappa^{} \,,
\end{equation}
where $T_{\mu\nu}^{}$ is the energy-momentum tensor of the theory,
\begin{equation}
 T_{\mu\nu}^{}~
 =~\partial_\mu^{} \varphi \; \partial_\nu^{} \varphi \, - \,
   \mathslf{g}_{\mu\nu}^{} L \,.
\end{equation}

Elucidating the relation between the Noether currents that appear in
our formulation and the energy-momentum tensor in a more general
context is one of the problems presently under investigation.

\section{Conclusions and Outlook}

In this paper, we have taken first steps towards a description of symmetries
in field theory using Lie groupoids and Lie algebroids, instead of the
traditional approach that uses Lie groups and Lie algebras but requires
infinite-dimensional ones as soon as local symmetries are involved. \linebreak
Our main motivation for doing so arises from the observation that in
relativistic field theories, the need to consider local symmetries is
almost unavoidable, since here the notion of a global symmetry~--
which, by definition, applies the same transformation at every point
of space-time~-- is a mathematical artifact without physical meaning.
After all, any physical implementation of such a requirement of rigidity
would violate the principle of space-time locality, according to which
no information can be exchanged between space-like separated regions
of space-time.
(This is really the same argument as the one showing that in
relativistic mechanics there is no such thing as a rigid body.)
We argue that the theory of Lie groupoids and Lie algebroids
provides the adequate mathematical machinery to describe
local symmetries in field theory; in particular, this applies
to gauge theories, whose geometric formulation using principal
bundles and connections has become standard wisdom during
the 1970s.
In this context, it may be amusing to note that Ehresmann already
invented all the essential mathematical notions (principal bundles,
connections \emph{and} Lie groupoids) during the 1950s, more or
less in one stroke, but mathematicians have for several decades
used only the first two and largely neglected the third, and
this lack of balance has proliferated into physics.%
\footnote{As a testimony to this statement, we may quote the
classical textbooks of Kobayahi and Nomizu.}
So in order to incorporate Lie groupoids and Lie algebroids
into the picture, we have to catch up on four decades of delay.

As an example of what can be gained, we may quote the classical
difficulties with unraveling the true symmetry of typical lagrangians
in field theory over curved space-times, provided we are interested
in including space-time symmetries.
In the traditional group-theoretical approach, the pertinent symmetry
is given by the isometry group of space-time (in special relativity,
the Poincar\'e group), which may collapse to a trivial group under
arbitrarily small perturbations of the metric.
Much of this instability disappears when we use groupoids: what
appears there is the orthonormal frame groupoid of space-time,
which is quite stable under arbitrary perturbations (even large ones),
but is generically non-holonomous.
Thus the notion of holonomous or non-holonomous subgroupoids
of jet groupoids, which has no analogue in traditional group theory,
appears to be an essential tool for understanding this issue.
We plan to elaborate further on this point in the second part
of this series.

\section*{Acknowledgements}

The work of the first and the last author has been done in partial
fulfillment of the requirements for the degree of Doctor in Science,
under the supervision of the second author, and has been supported
by fellowships from CAPES (Coordena\c{c}\~ao de Aperfei\c{c}oamento
de Pessoal de N\'{\i}vel Superior), Brazil. The second author
acknowledges partial financial support from CNPq (Conselho
Nacional de Desenvolvimento Cient\'{\i}fico e Tecnol\'ogico),
Brazil.

{\footnotesize

\end{document}